\documentclass[11pt]{article}
\usepackage{amsmath,amssymb,amsfonts,amsthm}
\usepackage[colorlinks,citecolor=blue,urlcolor=blue,linkcolor=red,]{hyperref}
\usepackage[margin=1.3in]{geometry}

 \usepackage{hyperref}
\usepackage{prodint}
\usepackage{dsfont}
\usepackage{mathbbol}
\usepackage{graphicx}
\DeclareSymbolFontAlphabet{\amsmathbb}{AMSb}%
\usepackage[round]{natbib}
\usepackage{subfigure}

\newcommand{\F}{{\amsmathbb{F}}}
\newcommand{\E}{{\amsmathbb{E}}}
\newcommand{\pr}{{\amsmathbb{P}}}
\newcommand*\dd{\mathop{}\!\mathrm{d}}

\newcommand{\dist}{{\:\stackrel{d}{\to}\:}}
\newcommand{\prob}{{\:\stackrel{\amsmathbb{P}}{\to}\:}}

\theoremstyle{plain}

\newtheorem{theorem}{Theorem}
\newtheorem{lemma}{Lemma}
\newtheorem{definition}{Definition}

\newtheorem{corollary}{Corollary}

\newtheorem{condition}{Condition}

\begin{document}
\title{Censored and extreme losses: functional convergence and applications to tail goodness-of-fit}
\author{Martin Bladt and Christoffer Øhlenschlæger\\
University of Copenhagen}
\date{\today}
\maketitle


\begin{abstract}
This paper establishes the functional convergence of the Extreme Nelson--Aalen and Extreme Kaplan--Meier estimators, which are designed to capture the heavy-tailed behaviour of censored losses. The resulting limit representations can be used to obtain the distributions of pathwise functionals with respect to the so-called tail process. For instance, we may recover the convergence of a censored Hill estimator, and we further investigate two goodness-of-fit statistics for the tail of the loss distribution. Using the the latter limit theorems, we propose two rules for selecting a suitable number of order statistics, both based on test statistics derived from the functional convergence results. The effectiveness of these selection rules is investigated through simulations and an application to a real dataset comprised of French motor insurance claim sizes.
 
\end{abstract}

\section{Introduction}

Investigating extreme events that involve right censoring is essential in insurance statistics, as ignoring these censoring effects can introduce significant bias. When extreme values are not fully observed due to censoring, treating them as complete data points can lead to a severe underestimation of their potential impacts. This issue commonly arises in the evaluation of large insurance claims, policyholder lifespans, natural disaster damages, and other risk factors, see for instance \cite{matthys2004estimating, reynkens2017modelling, goegebeur2024conditional}. Thus, understanding and analyzing censored extremes is crucial for managing these insurance risks effectively, even more so in the wake of the rising frequency and severity of catastrophic events globally.

This paper aims to establish the mathematical foundations for statistical inference on data exhibiting heavy tails and censoring, approached from a more general pathwise perspective. Thus, we extend the recent methodology of \cite{EKM} through an entirely new mathematical framework using empirical processes. The results in the latter reference have favorable finite-sample behavior, making their pathwise analysis all the more crucial. However, their results rely on methods from \cite{stute}, which work for integrals of the EKM estimator. The resulting theorems in the sequel impose conditions of a similar nature to those in \cite{einmahl2008statistics}, and are proved using techniques from empirical processes theory. By considering the pathwise behavior of EKM estimators and their distributional representations, we are able to recover established results for integrals of EKM estimators but also derive limit results for arbitrary functionals. This approach broadens the applicability of EKM methods and aims to establish them as the standard way of dealing with data with medium to low censoring proportions.

We provide a functional central limit theorem for extreme versions of the Nelson-Aalen and Kaplan-Meier estimators, which can subsequently be used to establish limit theorems for various functionals. As an illustrative example, we focus on goodness-of-fit testing, where extreme versions of the Kolmogorov--Smirnov and Cram\'er--von Mises statistics are employed to assess the quality of fit in the tail of the distribution. Intuitively, a better fit in the tail region leads to a more reliable estimation for any procedure that utilizes that exact number of upper-order statistics. Consequently, we specialize out approach even further to provide a finite-sample method for determining the optimal number of order statistics to consider in the estimation procedure, based on the limit theorems of the goodness-of-fit statistics. 

Selection rules have been studied extensively in the non-censored case, where several rules have been proposed. One review of selection rules can be found in \cite{ExtremeSELEC}. To mention two examples, we have \cite{Beirlantselect}, where they suggest a rule based on minimizing the asymptotic mean squared error and \cite{Thresholdselec}, where one of their suggestions is a rule based on integrated square error between the theoretical density and the parametric estimated density. In the censored case, for a different estimator, the paper \cite{Autoselect} applies a selection rule based on the information matrix test in the \emph{Peaks-Over-Thresholds} method. In \cite{bladt2021trimmed}, also for a different estimator, they suggest a selection rule based on minizing the mean squared error of a trimmed Hill estimator, with applications to third-party liability insurance. The crucial aspect of this paper is that we propose a selection rule for the Extreme Kaplan-Meier estimator, $\amsmathbb{F}_{k,n}$, as proposed in \cite{EKM}. The latter paper shows simulations studies where integrals of the type  $\int \varphi \dd \amsmathbb{F}_{k,n}$ outperform traditional estimators. The integrals, which they coin EKMI, contain multiple interesting tail-estimators. This further highlights the relevance of proposing tailor-made selection rules for the Extreme Kaplan-Meier estimator.

To evaluate the effectiveness of the results of this paper, we conduct simulation studies that demonstrate the practical utility of our asymptotic representations in finite samples. In particular, we provide evidence that the theoretical representations of our estimators are reasonable even when applied to medium datasets. Additionally, we apply our approach to a real-world data set, specifically analyzing claims from a French motor insurance company. Here, settlement times for claims can span several years, resulting in a naturally right-censored dataset at any given evaluation date (open claims being right-censored).

The paper is organized as follows: Section \ref{sec:setting} introduces the notation and general setup. Section \ref{sec:Con} presents results on the pathwise consistency and normality of the Extreme Nelson-Aalen and Extreme Kaplan-Meier estimators. Section \ref{sec:Pathwise} investigates the finite sample behavior of these estimators. Section \ref{sec: Applied} applies the obtained pathwise asymptotics, demonstrating the consistency and asymptotic normality of the Hill estimator and proposing two rules for selecting the number of order statistics. These rules are then evaluated through simulation and data from insurance claims.

\section{Setting}
\label{sec:setting}

Let $\{X_i\}_{i=1}^n$ and  $\{Y_i\}_{i=1}^n$ be independent sequences of independent and identical distributed (i.i.d.) random variables with distribution functions $F_X$, $F_Y$, respectively. Here, the target distribution function if $F_X$. We assume we are in a right-censoring regime, such that we only observe $\{(Z_i,\delta_i)\}_{i=1}^n$, which are given by $Z_i=\min\{X_i,Y_i\}$ and $\delta_i=I(X_i\leq Y_i)$. The distribution function of $Z_i$ is denoted $F_Z$.

Throughout the paper, we consider the case where both $F_X$ and $F_Y$ are heavy-tailed, and more precisely, belong to the Frech\'et domain of attraction. In particular, the risks of interest have tails that exhibit power-law behavior. This regime has been recently studied in different settings, for instance, \cite{stupfler2016estimating,bladt2021trimmed,goegebeur2024conditional}. Thus, we have that
\begin{align*}
    1-F_i(x)
    = 
     L_i(x) \ x^{-\frac{1}{\gamma_i}},
\end{align*}
for $i\in \{X,Y\}$, where $\gamma_i$ is named the tail-index, and where $L_i$ is slowly varying at infinity, that is $\lim_{t\to\infty}L_i(tx)/L(t)=1$ for every $t>0$. This implies that the corresponding tail-counterparts, $F^t_X$ and $F^t_Y$, satisfy  
\begin{align*}
    1-F^t_i(x):=\frac{1-F(tx)}{1-F(t)}\to x^{-1/\gamma_i} \quad \text{as } t\to \infty.
\end{align*}
 A further consequence is that $Z_i$ likewise belongs to the Frech\'et domain of attraction with tail index $\gamma=\frac{\gamma_Y\gamma_X}{\gamma_Y+\gamma_X}$. A recurring expression in the sequel is the limiting tail distribution of $Z$, and hence to ease notation we denote it by $T(s)=s^{-1/\gamma}$. Finally, let $p=\frac{\gamma_Y}{\gamma_X+\gamma_Y}$, which is the asymptotic frequency of non-censoring in the tail.
 
Let $u_n=F_Z^{\leftarrow}(1-k/n)$ be a sequence of thresholds where $k=k_n$ is an intermediate sequence satisfying $k_n\to\infty$ and $k_n/n\to0$.  Large $k$ is necessary for having sufficient data to perform estimation, while small $k/n$ indicates the need to focus solely on the tail.

We require the following assumption, which states that the asymptotic censoring proportion adequately stabilizes at $p$, confer also with the assumptions of Theorem 1 in \cite{einmahl2008statistics}:  
\begin{condition}\label{cond1}
Let $p(s)=\pr(\delta_1=1|Z_1>s)$. The following convergence holds:
\begin{align*}
\sqrt{k}(p(u_n)-p)\to0.
\end{align*}
Further, for any interval $[a,b]\subset \amsmathbb{R}_+$
\begin{align*}
\limsup_{n\to\infty}\sqrt{k}\sup_{a\le s\le t<b}\Big|\frac{\pr(u_nt\ge Z_1>u_ns,\,\delta_1=1)}{1-F_Z(u_n)}-p\left(s^{-1/\gamma}-t^{-1/\gamma}\right)\Big|=0.
\end{align*}
\end{condition}
Notice that without the $\sqrt{k}$, the second part of the latter condition is satisfied automatically by regular variation in case $p=1$, that is when there is no censoring present. Here, the interplay between $k$ and $u_n$ is key.

In this paper, we aim to investigate the \emph{pathwise distribution} of the Extreme Kaplan-Meier (EKM) estimator, which was first proposed in \cite{EKM}. It is given by
\begin{align}
    \amsmathbb{F}_{k,n}(x)=1-\prod_{i=1}^k \left[ 1-\frac{\delta_{[n-i+1:k]}}{i}\right]^{I(Z_{n-i+1,n}/Z_{n-k,n}\leq x)}.
    \label{eq:EKM}
\end{align} 
where $Z_{1,n}\leq Z_{2,n}\leq \ldots, \leq Z_{n,n}$ are the order statistics of the censored observations and $\delta_{[i:n]}$ are the concomitants of $\{Z_{i,n}\}_{i=1}^n$. Intuitively speaking, $\amsmathbb{F}_{k,n}$ removes the censoring effects while still zooming into the tail. The pathwise properties of such an estimator then shed light on the quality of the convergence to the asymptotic Pareto tail of $F_X$. To proceed with the investigation of the EKM estimator, we use an equivalent, but different, expression for it which simplifies its analysis using empirical process techniques. First, we define the tail empirical distribution function.

\begin{definition}[Tail empirical distribution function with random levels]
Define the tail empirical distribution and subdistribution functions with random levels, for $s>0$, as
\begin{align*}
&\widehat T_n(s)=\frac{1}{k}\sum_{i=1}^nI(Z_i>Z_{n,n-k} s),\\ &\widehat T^1_n(s)=\frac{1}{k}\sum_{i=1}^nI(Z_i>Z_{n,n-k} s,\,\delta_i=1).
\end{align*}
Define also
\begin{align*}
 \widehat{\amsmathbb{T}}_n(s)=\sqrt{k}(\widehat{T}_n(s)-T(s)),\quad  \widehat{\amsmathbb{T}}^1_n(s)=\sqrt{k}(\widehat{T}^1_n(s)-pT(s)).
\end{align*}
\end{definition}
Now we are in position to define the Extreme Nelson-Aalen estimator and thereafter redefine the EKM estimator. The latter redefinition follows from a well-known property of the product integral (cf. \cite{gill1990survey}), for which we use the symbol $\prodi$. 
\begin{definition}[Extreme Nelson-Aalen and Kaplan-Meier estimators]
We define the extreme versions of the Nelson-Aalen and Kaplan-Meier estimators as follows:
\begin{align*}
\mathbb{\Lambda}_{k,n}(t)&=-\int_1^t\frac{1}{\widehat{\amsmathbb{T}}_n(s-)}\dd \widehat{\amsmathbb{T}}^1_n(s),\\
\amsmathbb{F}_{k,n}(t)&=1-\Prodi_1^t(1-\dd \mathbb{\Lambda}_n(s)).
\end{align*}
Define also $\Lambda^\circ(s)=\frac{1}{\gamma_X}\log(s)$ and $F^\circ(s)=1-s^{-1/\gamma_X}$.
\end{definition}
The idea now is to use the pathwise weak convergence of $\amsmathbb{F}_{k,n}$ to $F^\circ$ to assess the quality of asymptotic approximations of functionals of $\amsmathbb{F}_{k,n}$. For instance, we may retrieve straightforwardly an estimator for $\gamma_X=\int \log \dd F^\circ$ by considering $\int \log \dd \amsmathbb{F}_{k,n}$. However, the degree to which we can trust this approximation depends on how close the two integrator paths are, which in turn can be qualified in terms of their pathwise asymptotic Gaussian process representations.

\section{Convergence of tail processes}
\label{sec:Con}
To be able to show the pathwise convergence of EKM estimtor, we first require to do some groundwork, for which an approach similar to \cite{kulik2020heavy} is
the most direct route. Some parts of the proofs follow directly the latter reference and are therefore omitted, thus concentrating mainly on the differences required for our setting.
Hence, we begin by looking at the tail empirical distribution functions.

\begin{definition}[Tail empirical distribution functions]
Define the tail empirical distribution and sub-distribution functions for $ s>0$ as
\begin{align*}
\widetilde T_n(s)&=\frac{1}{k}\sum_{i=1}^nI(Z_i>u_n s),\\
\widetilde T^1_n(s)&=\frac{1}{k}\sum_{i=1}^nI(Z_i>u_n s,\,\delta_i=1).
\end{align*}
and its expectations
\begin{align*}
T_n(s)&=\E[\widetilde T_n(s)]=\frac{1-F_Z(u_n s)}{1-F_Z(u_n)},\\
T^1_n(s)&=\E[\widetilde T^1_n(s)]=\frac{\pr(Z>u_n s,\,\delta_1=1)}{1-F_Z(u_n)}.
\end{align*}
Define also $T(s)=s^{-1/\gamma}$ and $p=\frac{\gamma_Y}{\gamma_X+\gamma_Y}$.
\end{definition}
The following theorem shows weak consistency of the tail empirical functions.
\begin{theorem}[Consistency of the tail empirical function]
\label{thm:prob}
We have that
\begin{align*}
\widetilde T_n(s)\prob  T(s),\quad \widetilde T^1_n(s)\prob  pT(s),
\end{align*}
uniformly on bounded intervals. Moreover $Z_{n,n-k}/u_n\prob 1$.
\end{theorem}
\begin{proof}
The first and last assertions are well-known. We prove the second one. Fix $s>0$. Then
\begin{align*}
\widetilde T^1_n(s)-pT(s)=[\widetilde T^1_n(s)- T^1_n(s)]+[T^1_n(s)-pT(s)].
\end{align*}
The first term is $o_\pr(1)$ by the LLN, and for the second one we obtain
\begin{align*}
T^1_n(s)-pT(s)&=\frac{\pr(Z>u_n s,\,\delta_1=1)}{1-F_Z(u_n)}-ps^{-1/\gamma}\\
&=\pr(\delta_1=1|Z>u_n)\frac{1-F_Z(u_n s)}{1-F_Z(u_n)}-ps^{-1/\gamma},
\end{align*}
which by regular variation is also $o_\pr(1)$. By Dini’s Theorem, the convergence is uniform on bounded intervals.
\end{proof}
We now introduce the tail empirical processes and show their joint weak convergence.
\begin{definition}[Tail empirical process]
Define for a sequence $u_n$ of thresholds the tail empirical process of the distribution and sub-distribution functions for $s>0$ by
\begin{align*}
\widetilde {\amsmathbb{T}}_n(s)=\sqrt{k}(\widetilde T_n(s)-T_n(s)),\quad \widetilde {\amsmathbb{T}}^1_n(s)=\sqrt{k}(\widetilde T^1_n(s)-T^1_n(s)).
\end{align*}
Define also
\begin{align*}
{\amsmathbb{T}}_n(s)=\sqrt{k}(\widetilde T_n(s)-T(s)),\quad{\amsmathbb{T}}^1_n(s)=\sqrt{k}(\widetilde T^1_n(s)-pT(s)).
\end{align*}
\end{definition}
\begin{theorem}[Joint weak convergence of the tail empirical process]
On compacts, 
\begin{align*}
(\widetilde {\amsmathbb{T}}_n,\,\widetilde {\amsmathbb{T}}^1_n)\dist  ({\amsmathbb{T}},\,{\amsmathbb{T}}^1)
\end{align*}
where $(\amsmathbb{T},\,\amsmathbb{T}^1)=(\amsmathbb{W}\circ T,\,\sqrt{p}\cdot\amsmathbb{W}^1\circ T)$ and $(\amsmathbb{W},\,\amsmathbb{W}^1)$ is a bivariate Brownian motion with correlation $\rho=\sqrt{p}$. In particular, one may choose $$\amsmathbb{W}^1=\sqrt{p} \,\amsmathbb{W}+\sqrt{1-p}\,\amsmathbb{W}^\perp$$
with $\amsmathbb{W},\,\amsmathbb{W}^\perp$  independent standard Brownian motions.
\end{theorem}
\begin{proof}
Let us consider the finite-dimensional distributions by applying the Lindenberg CLT. Define
\begin{align*}
&\xi_{n,i}(s,\alpha,\beta)=\\
&k^{-1/2}\{\alpha I(Z_i>u_ns)+\beta I(Z_i>u_ns,\,\delta_i=1)-\alpha\pr(Z_i>u_ns)-\beta\pr(Z_i>u_ns,\,\delta_i=1)\}.
\end{align*}
so that $\sum_{i=1}^n \xi_{n,i}(\cdot,\alpha,\beta)=\alpha\widetilde {\amsmathbb{T}}_n+\beta\widetilde {\amsmathbb{T}}^1_n$. Observe that by regular variation:
\begin{align*}
&n\mbox{Cov}(\xi_{n,i}(s,\alpha_1,\beta_1),\xi_{n,i}(t,\alpha_2,\beta_2))\\
&=\frac{\alpha_1\alpha_2\pr(Z_1>u_n(s\vee t))+(\alpha_1\beta_2+\beta_1(\alpha_2+\beta_2))\pr(Z_1>u_n(s\vee t),\,\delta_1=1)}{1-F_Z(u_n)}+o(1)\\
&\to \big(\alpha_1\alpha_2+p(\alpha_1\beta_2+\beta_1(\alpha_2+\beta_2))\big)(s\vee t)^{-1/\gamma}.
\end{align*}
Next, by observing that $|\xi_{n,i}(s,\alpha,\beta)|\le(\alpha+\beta)(n(1-F(u_n)))^{-1/2}\to0$ deterministically, the following condition holds for any $\varepsilon>0$.
\begin{align*}
n\E[\xi_{n,i}(s,\alpha,\beta)^2I(|\xi_{n,i}(s,\alpha,\beta)|>\varepsilon)]\to0.
\end{align*}
This shows, by Lindenberg's CLT and by the Cram\'er-Wold device normality of the finite-dimensional distributions, jointly for both processes. The covariance structure can be calculated considering different choices of $\alpha_i,\beta_i$ in the above formulae.

To prove functional convergence, we make use of the Bracketing CLT. Recall that functional convergence of each individual process together with finite-dimensional distributional convergence implies joint functional convergence. Thus we only show functional convergence of $\widetilde {\amsmathbb{T}}^1_n$ (the other process being simpler).  In this case let $Z_{n,i}=k^{-1/2}I(Z_i>u_n\:\cdot\,,\, \delta_i=1)$. Clearly the argument for the first condition of the Bracketing CLT is similar to the one presented above (notice that the bound is uniform), and thus is easily satisfied. Next, let us find the bracketing number. 

Let $\varepsilon>0$, and $[a,b]$ be an interval with $a<b$. Define $$A_n(t)=n\E[	(Z_{n,i}(a)-Z_{n,i}(t))(Z_{n,i}(a)-Z_{n,i}(b))].$$ This is a right-continuous with left limits, non-decreasing function which by Condition \ref{cond1} satisfies $A_n(0)=0$, and as $n\to\infty$
\begin{align*}
A_n(b)=\frac{\pr(u_nb\ge Z_i>u_na,\,\delta_i=1)}{1-F(u_n)}\to p\left(a^{-1/\gamma}-b^{-1/\gamma}\right).
\end{align*}
Let $a=t_1<\dots<t_N=b$ be a partition such that $A_n(t_{i}-)-A_n(t_{i-1})<\varepsilon^2$ for every $i$. Then 
\begin{align*}
&n\cdot \E\left[\sup_{s,t\in [t_{i-i},t_i)}(Z_{n,i}(s)-Z_{n,i}(t))^2\right]\\
&\le n\E[(Z_{n,i}(t_i-)-Z_{n,i}(t_{i-1}))(Z_{n,i}(a)-Z_{n,i}(b))]<\varepsilon^2.
\end{align*}
By Condition \ref{cond1} we have, for $n\ge n_0$ for some large enough $n_0$, at most $N(\varepsilon,a,b):=\frac{2p(a^{-1/\gamma}-b^{-1/\gamma})}{\epsilon^2}$ terms in any partition, independently of $n$. Thus it only remains to verify the entropy condition,
which follows from the fact that $N_{[\:]}(\varepsilon,\mathcal{F},L_2(\pr))\le N(\varepsilon,a,b)$ and 
$\int_0^1 \sqrt{\log N(\varepsilon,a,b)}\dd\varepsilon<\infty$, so that clearly $$\int_0^{\delta_n}\sqrt{\log N_{[\:]}(\varepsilon,\mathcal{F},L_2(\pr))}\dd\varepsilon\to0\quad \mbox{for any}\quad \delta_n\downarrow0.$$ This concludes the proof.
\end{proof}

We now introduce the empirical quantile process, which is an important building block for showing weak convergence of the tail empirical process with random levels but is also interesting in its own right. 
\begin{definition}[Empirical quantile process]
The empirical quantile process is defined by
\begin{align*}
\widetilde Q_{n}(t)=\widetilde T_n^{\leftarrow}([kt]/k)=\frac{Z_{n,n-[kt]}}{u_n}, \quad s\in(0,n/k).
\end{align*}
Setting $Q(t)=t^{-\gamma}$ we define the associated empirical process as
\begin{align*}
\amsmathbb{Q}_n=\sqrt{k}(\widetilde Q_n-Q),\quad s\in(0,n/k).
\end{align*}
\end{definition}

Define $B_n(s_0)=\sup_{t\ge s_0}|T_n(t)-T(t)|$, which converges uniformly  to $0$ for each $s_0>0$.
\begin{theorem}[Joint convergence of tail empirical and quantile processes]
\label{thm:quan}
Let $k=k_n$ grow slow enough such that for some $s_0\in(0,1)$, $\sqrt{k}B_n(s_0)\to0$. Then
\begin{align*}
(\amsmathbb{T}_n,\,\amsmathbb{T}^1_n,\, \amsmathbb{Q}_n)\dist(\amsmathbb{T},\,\amsmathbb{T}^1,\, -Q'\cdot\amsmathbb{T}\circ Q)=(\amsmathbb{W}\circ T,\,\sqrt{p}\cdot\amsmathbb{W}^1\circ T,\:-Q'\amsmathbb{W})
\end{align*}
on $\ell^\infty([s_0,\infty))\times\ell^\infty([s_0,\infty))\times\ell^\infty((0,Q(s_0)])$. In particular,
\begin{align*}
(\amsmathbb{T}_n,\,\amsmathbb{T}^1_n,\, \sqrt{k}(Z_{n,n-k}/u_n-1))&\dist(\amsmathbb{T},\,\amsmathbb{T}^1,\, \gamma\cdot\amsmathbb{T}(1))\\
&=(\amsmathbb{W}\circ T,\,\sqrt{p}\cdot\amsmathbb{W}^1\circ T,\,\gamma\cdot\amsmathbb{W}(1)).
\end{align*}
\end{theorem}
\begin{proof}
We may write, on $[s_0,\infty)$, by assumption and using Condition \ref{cond1}, $(\amsmathbb{T}_n,\,\amsmathbb{T}^1_n)=(\widetilde{\amsmathbb{T}}_n,\,\widetilde{\amsmathbb{T}}^1_n)+o(1)$, and the latter converges to $(\amsmathbb{T},\,\amsmathbb{T}^1)$ in $\ell^\infty([s_0,\infty)^2)$. Then the desired joint convergence follows from an application of Vervaat’s Lemma (the continuous mapping theorem for the special case of the joint identity and inverse mappings) applied to the first coordinate process. 
\end{proof}

We are now in a position to leverage the previous results to finally assert weak convergence of the tail empirical process with random levels. The random levels are of course needed in practice since the deterministic levels are theoretical quantities which are not available from data.
\begin{theorem}[Weak convergence of tail empirical process with random levels]
Let $k=k_n$ grow slow enough such that for some $s_0\in(0,1)$, $\sqrt{k}B_n(s_0)\to0$. Then
\begin{align*}
(\widehat{\amsmathbb{T}}_n,\, \widehat{\amsmathbb{T}}^1_n)&\dist(\amsmathbb{T}-T\cdot\amsmathbb{T}(1),\,p\cdot\amsmathbb{T}-p\,T\cdot\amsmathbb{T}(1)-\sqrt{p(1-p)}\amsmathbb{T}^\perp)\\
&=(\amsmathbb{B}\circ T,\,p\cdot \amsmathbb{B}\circ T-\sqrt{p(1-p)}\amsmathbb{W}^\perp\circ T)
\end{align*}
in $\ell^\infty([s_0,\infty)^2)$, where $\amsmathbb{B}$ is standard Brownian Bridge, independent of $\amsmathbb{W}^\perp$. In particular, the second coordinate process has covariance $p(s\wedge t)-p^2st$.
\end{theorem}

\begin{proof}
Write $r_n=Z_{n,n-k}/u_n$, and 
\begin{align*}
&(\widehat{\amsmathbb{T}}_n,\,\widehat{\amsmathbb{T}}^1_n)\\
&=\sqrt{k}\big((\widetilde T_n(\cdot\, r_n),\,\widetilde T^1_n(\cdot\, r_n))-(T_n(\cdot\, r_n),\,T^1_n(\cdot\, r_n))\big)\\
&\quad+\sqrt{k}\big( (T_n(\cdot\, r_n),\,T^1_n(\cdot\, r_n))-(T(\cdot\, r_n),pT(\cdot\, r_n))\big)\\
&\quad+\sqrt{k}\big( (T(\cdot\, r_n),\,pT(\cdot\, r_n))-(T,\,pT)\big).
\end{align*}
Since $r_n\prob 1$, the first term has pathwise convergence to $(\amsmathbb{T},\,\sqrt{p}\, \amsmathbb{T}^1)$ by Slutsky's lemma. By the same convergence
\begin{align*}
\pr\big(||(T_n(s r_n),\,T^1_n(s r_n))-(T(s r_n),\,pT(s r_n))||\le B_n(s_0)\big)\to1
\end{align*} 
for any $s\ge s_0$, so by assumption the second term becomes pathwise $o_{\pr}(1).$ The delta method applied to $x\mapsto (T(xs),\,pT(xs))$ yields weak pathwise convergence of the last term (jointly with the first one), to \begin{align*}
\big(\frac{-1}{\gamma}s^{-1/\gamma}\gamma \amsmathbb{T}(1),\,\frac{-p}{\gamma}s^{-1/\gamma}\gamma \amsmathbb{T}^1(1)\big)=-T(s)(\amsmathbb{T}(1),\,p\,\amsmathbb{T}(1)).
\end{align*}
Collecting the terms yields the conclusion of the theorem.
\end{proof}

\section{Pathwise convergence of ENA and EKM estimators}
\label{sec:Pathwise}
This section provides pathwise consistency and normality of the ENA and EKM estimators. We refer the reader to \cite{vdv2000} for the main weak convergence arguments required in the proofs.  We also consider their finite-sample behavior through Monte Carlo simulation, which empirically verifies the quality of the asymptotic expressions.

\subsection{Consistency and weak convergence}

\begin{theorem}[Consistency of ENA and EKM]
    \label{thm:consistency}
    The following convergence holds uniformly on compacts:
    \begin{align*}
    &\mathbb{\Lambda}_{k,n}\prob\Lambda^\circ\\
        &\amsmathbb{F}_{k,n}\prob F^\circ.
    \end{align*}
\end{theorem}

\begin{proof}
    It is straightforward to see that by random censoring (and recalling that $\gamma=\gamma_X\gamma_Y/(\gamma_X+\gamma_Y)$ and $p=\gamma/\gamma_X$), pointwise, $\mathbb{\Lambda}_n(t)\prob\Lambda^\circ(t)$. The convergence is easily extended by integration by parts to $\amsmathbb{F}_n\prob F^\circ(s)$, and uniformly on compacts by Dini's theorem.
\end{proof}

For establishing weak convergence, we again need to ensure that $k$ does not grow too quickly. More precisely, we assume that the growth depends on $B_n(s_0):=\sup_{t\ge s_0}|T_n(t)-T(t)|$, where recall that $T_n(s)=\frac{1-F_Z(u_n s)}{1-F_Z(u_n)}$. This is akin to Condition \ref{cond1}.
\begin{theorem}[Weak convergence of the ENA and EKM estimators]
\label{thm:normality}
Let $k=k_n$ be such that for some $s_0\in(0,1)$, the bias term $\sqrt{k}B_n(s_0)\to0$. Then, on bounded intervals, 
\begin{align*}
\sqrt{k}(\mathbb{\Lambda}_{k,n}-\Lambda^\circ)&\dist \frac{p\,\amsmathbb{B}\circ T+\sqrt{p(1-p)}\,\amsmathbb{B}^\perp\circ T}{T}-\sqrt{p(1-p)}\int_{T(\cdot)}^1 \frac{\,\amsmathbb{W}^\perp(u)}{u^2}\dd u\\
&\quad =:\amsmathbb{Z}\circ T\\
\sqrt{k}(\amsmathbb{F}_{k,n}-F^\circ)&\dist  (1-F^\circ)\, \amsmathbb{Z}\circ T,
\end{align*}
where $\amsmathbb{W},\,\amsmathbb{W}^\perp$ are standard brownian motions, and where $\amsmathbb{B},\,\amsmathbb{B}^\perp$ are brownian bridges constructed from them. In particular, $\amsmathbb{Z}$ is a Gaussian martingale.
\end{theorem}
It can be noted that the theorem collapses to the standard case if there is no censoring, i.e. if $p=1$.

\begin{proof}
We apply the functional delta theorem twice. To this end define the maps
\begin{align*}
\phi_1(f_1,f_2)&=\int_1^\cdot\frac{1}{f_1(s-)}\dd f_2(s)\\
\phi_2(f)&=\Prodi_1^\cdot(1-\dd f(s)),
\end{align*}
with their corresponding Hadamard derivatives given by
\begin{align*}
\phi'_1(f_1,f_2)(\alpha_1,\alpha_2)&=\int_1^\cdot\frac{1}{f_2(s)}\dd \alpha_1(s)-\int_0^\cdot \frac{\alpha_2(s)}{f_2(s)^2}\dd f_1(s)\\
\phi'_2(f)(\alpha)&=\int_1^\cdot \Prodi_{(1,s]}(1-\dd f(s)) \Prodi_{(s,\cdot]}(1-\dd f(s)) \dd \alpha(s).
\end{align*}
By a first application of the functional delta method we obtain
\begin{align*}
&\sqrt{k}(\mathbb{\Lambda}_n-\Lambda^\circ)\dist -\int_1^\cdot \frac{1}{T(s)}\dd (p\,\amsmathbb{B}\circ T-\sqrt{p(1-p)}\,\amsmathbb{W}^\perp \circ T)(s)\\
&\quad+\int_1^t\frac{(\amsmathbb{B}\circ T)(s)}{T(s)^2}\dd(pT)(s)\\
&=-p\left[\frac{\amsmathbb{B}\circ T}{T}-\amsmathbb{B}(1)\right]+p\int_1^\cdot (\amsmathbb{B}\circ T)(s)\dd (\frac{1}{T})(s)\\
&\quad+\sqrt{p(1-p)}\left[\frac{\amsmathbb{W}^\perp\circ T}{T}-\amsmathbb{W}^\perp(1)\right]-\sqrt{p(1-p)}\int_1^\cdot (\amsmathbb{W}^\perp\circ T)(s)\dd (\frac{1}{T})(s)\\
&\quad+\int_1^\cdot\frac{(\amsmathbb{B}\circ T)(s)}{T(s)^2}\dd(pT)(s)\\
&=-p\left[\frac{\amsmathbb{B}\circ T}{T}-\amsmathbb{B}(1)\right]+\sqrt{p(1-p)}\left[\frac{\amsmathbb{W}^\perp\circ T}{T}-\amsmathbb{W}^\perp(1)\right]\\
&\quad+\sqrt{p(1-p)}\int_1^{T(\cdot)} \frac{\amsmathbb{W}^\perp(u)}{u^2}\dd u.
\end{align*}
A second application now yields
\begin{align*}
\sqrt{k}(\amsmathbb{F}_{k,n}-F^\circ)&\dist -\int_1^\cdot \Prodi_{(1,s]}(1-\dd \Lambda^\circ(s)) \Prodi_{(s,\cdot]}(1-\dd \Lambda^\circ(s)) \dd (\amsmathbb{Z}\circ T)(s)\\
&\quad=-\int_1^\cdot (1-F^\circ(s-)) \frac{1-F^\circ(\cdot)}{1-F^\circ(s)} \dd (\amsmathbb{Z}\circ T)(s)\\
&\quad=-\int_1^\cdot (1-F^\circ(\cdot)) \dd (\amsmathbb{Z}\circ T)(s). 
\end{align*}
\end{proof}

\begin{lemma}[Computation of the covariance of $\amsmathbb{Z}$]\label{lemma:computation}
We have that 
\begin{align*}
\mbox{Cov}(\amsmathbb{Z}(s),\amsmathbb{Z}(t))=p(s^{-1}\wedge t^{-1}-1).
\end{align*}
\end{lemma}
\begin{proof}
We have that 
\begin{align*}
\mbox{Cov}(\amsmathbb{Z}(s),\amsmathbb{Z}(t))&= p^2(s^{-1}\wedge t^{-1}-1)+ p(1-p)(s^{-1}\wedge t^{-1}-1)\\
&\quad + p(1-p)I_1-\frac{1}{s}p(1-p)I_2-\frac{1}{t}p(1-p)I_3.
\end{align*}
Assume $t\ge s$. Then
\begin{align*}
&I_1=\int_t^1 \int_s^1 \frac{x\wedge y}{x^2y^2}\dd x\dd y=\frac{t\log(ts)+\log(t/s)-2t+2}{t}.
\end{align*}
Regarding the second integral,
\begin{align*}
I_2=\int_t^1\frac{s\wedge u-su}{u^2}\dd u=s(\log(t)-1)+s/t,
\end{align*}
and finally
\begin{align*}
I_3=\int_s^1\frac{t\wedge u-tu}{u^2}\dd u=t\log(s)+\log(t/s)-t+1.
\end{align*}
Collecting terms yields
\begin{align*}
\mbox{Cov}(\amsmathbb{Z}(s),\amsmathbb{Z}(t))&= p(s^{-1}\wedge t^{-1}-1).
\end{align*}
and by symmetry, the $s\ge t$ case yields the same formula.
\end{proof}
It is easy to verify by Taylor expansion that for $t\in(0,1)$, the domain of $\amsmathbb{Z}$, the resulting variance is always positive, as expected:
\begin{align*}
\mbox{Var}(\amsmathbb{Z}(t))=p\Big(\frac{1}{t}-1\Big)\ge 0.
\end{align*}

\subsection{Finite sample behavior}
\label{sec:ver}
We now investigate the finite sample behavior of $\sqrt{k}(\mathbb{\Lambda}_{k,n}-\Lambda^\circ)$ and $\sqrt{k}(\F_{k,n}-F^\circ)$. We present the finite behavior when the data is Fr\'echet, Burr exact Pareto.

In the four left panels of Figure \ref{fig:Ver-fre} we see the mean and the variance of $\sqrt{k}(\mathbb{\Lambda}_{k,n}-\Lambda^\circ)$ at times $s=2$ and $s=4$. The solid line is based on a sample size of $n=1,000$ and the dotted line is for when $n=10,000$. The number of simulations is 500. We observe, that the mean vanishes as $k$ decreases, which aligns with the theory. Likewise, we see that the difference between the empirical and theoretical variance vanishes as $k$ decreases. However, we observe that it becomes unstable when $k$ is too small. This phenomenon is more apparent for $n=1,000$ than $n=10,000$, and highlights the importance of the intermediate sequence satisfying $k\to\infty$ and $k/n\to0$ at the same time. We also naturally observe that the estimation becomes more stable as $n$ increases. Note that $\mbox{Var}(\mathbb{Z}\circ T(2))=3.7$, $\mbox{Var}(\mathbb{Z}\circ T(4))=24.8$, which is the reason for the large difference in magnitude between the plots at time $s=2$ and $s=4$. In the four right panels of Figure \ref{fig:Ver-fre} we provide the same plot but now for $\sqrt{k}(\F_{k,n}-F^\circ)$. The conclusions are similar, except for a slight improvement of variance, and a slight growth in bias.

Analogous analysis and conclusions may be drawn from the exact Pareto and Burr distributions, which are provided in Figures \ref{fig:Ver-burr} and \ref{fig:Ver-par}. The latter distribution has the least amount of bias, and thus has the best finite-sample behavior, while the latter has even more severe bias than the Fr\'echet case. The bias-variance tradeoff is a common theme in statistics for extremes, and paves a promising line of research in our setting, since bias-reduced estimators, in the spirit of \cite{BEIRLANT2018114}, may be constructed for the ENA and EKM estimators, taking our pathwise expressions as a starting point.

\begin{figure}[htbp!]
    \centering
        \includegraphics[width=0.49\textwidth]{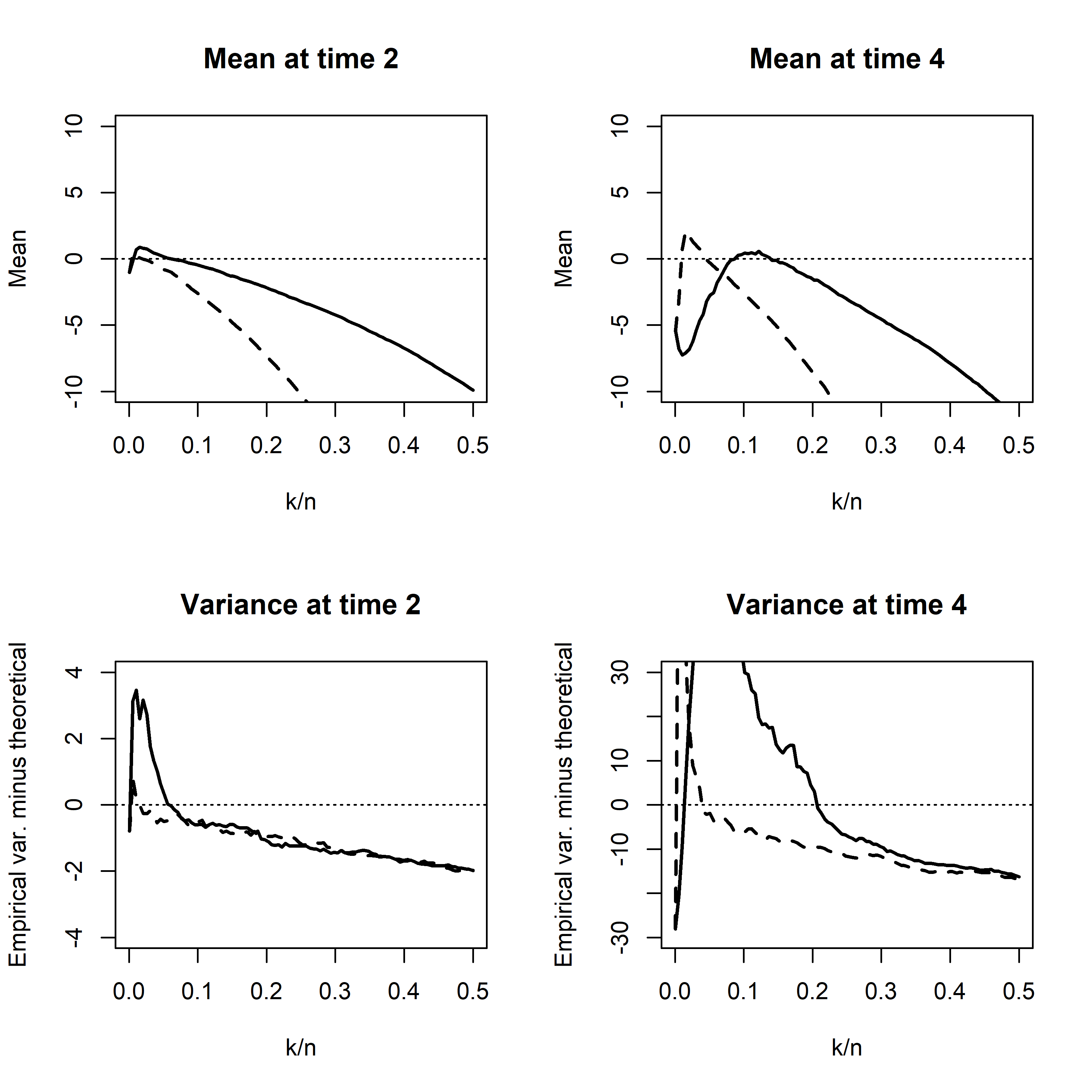}
                \includegraphics[width=0.49\textwidth]{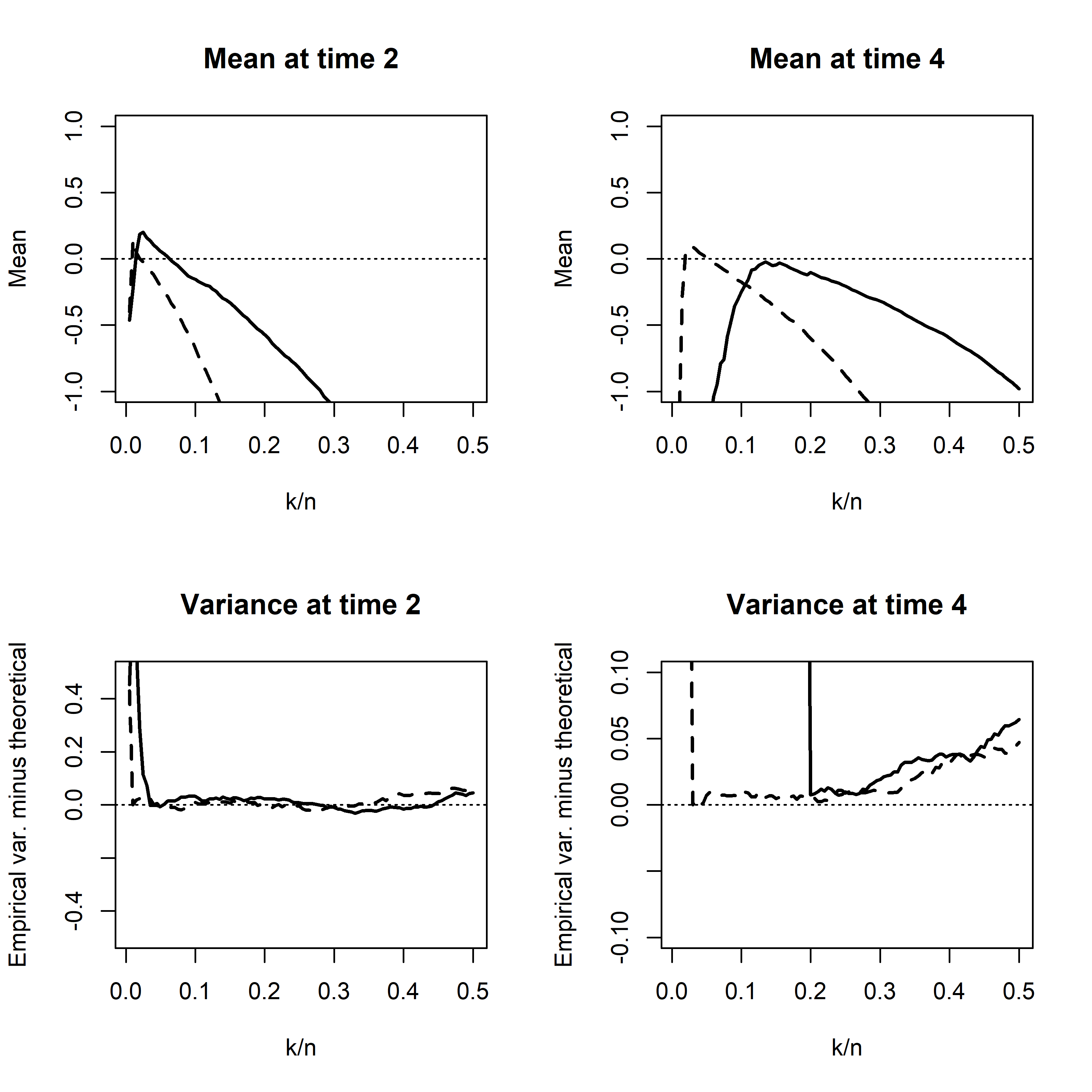}
    \caption{Finite sample behavior of $\sqrt{k}(\mathbb{\Lambda}_{k,n}-\Lambda^\circ)$ (four left panels) and $\sqrt{k}(\F_{k,n}-F^\circ)$ (four right panels). The sample sizes are $n=10,000$ (dashed line) and $n=1,000$ (solid line). The distribution is Fr\'echet with $\gamma_X=0.5$ and $\gamma_Y=1.5$, and we have $\mbox{Var}(\mathbb{Z}\circ T(2))=3.7$, $\mbox{Var}(\mathbb{Z}\circ T(4))=24.8$, $(1-F^\circ)^2(2)\mbox{Var}(\mathbb{Z}\circ T(2))=0.2$, $(1-F^\circ)^2(4)\mbox{Var}(\mathbb{Z}\circ T(4))=0.1$. } 
    \label{fig:Ver-fre}
\end{figure}

\begin{figure}[htbp!]
    \centering
            \includegraphics[width=0.49\textwidth]{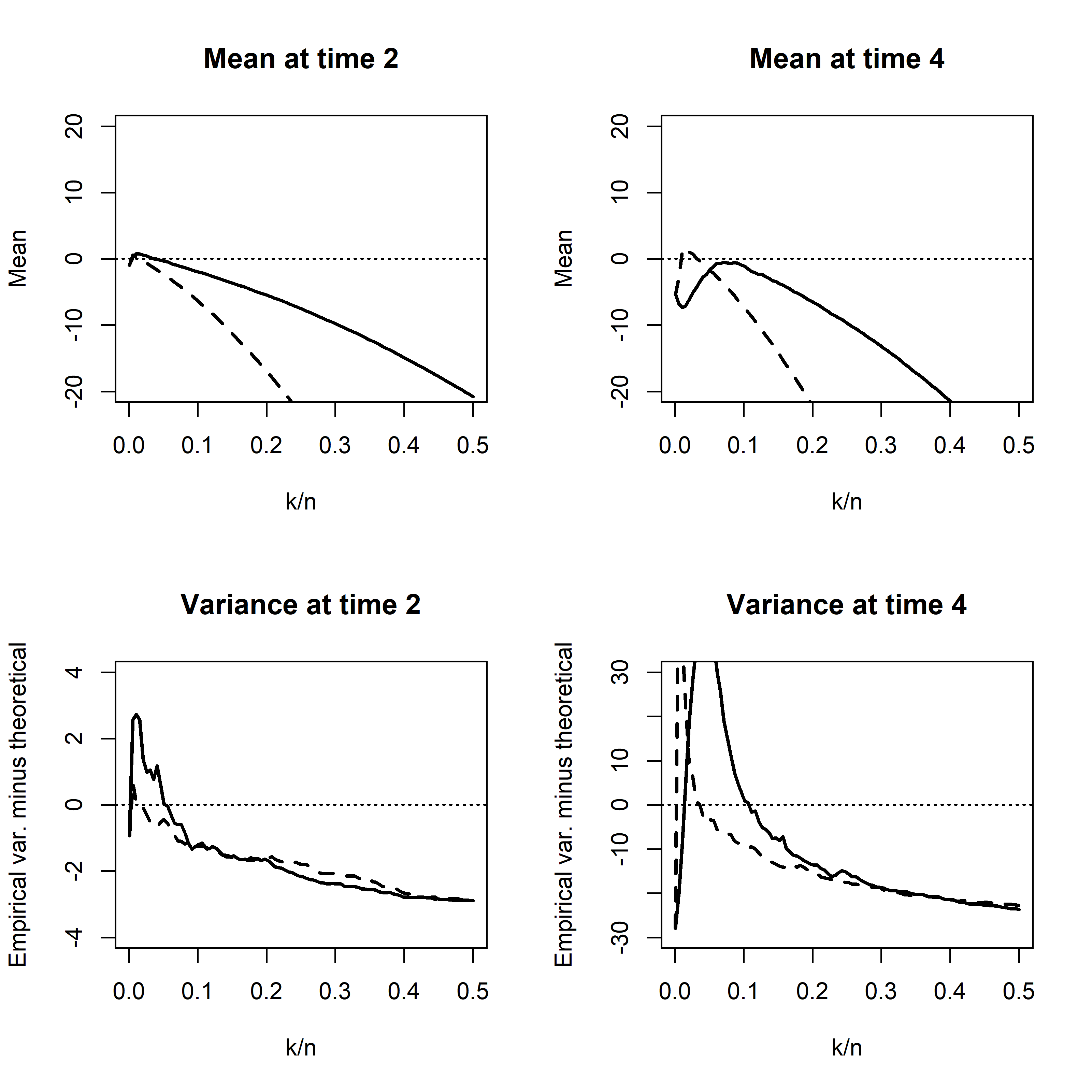}
        \includegraphics[width=0.49\textwidth]{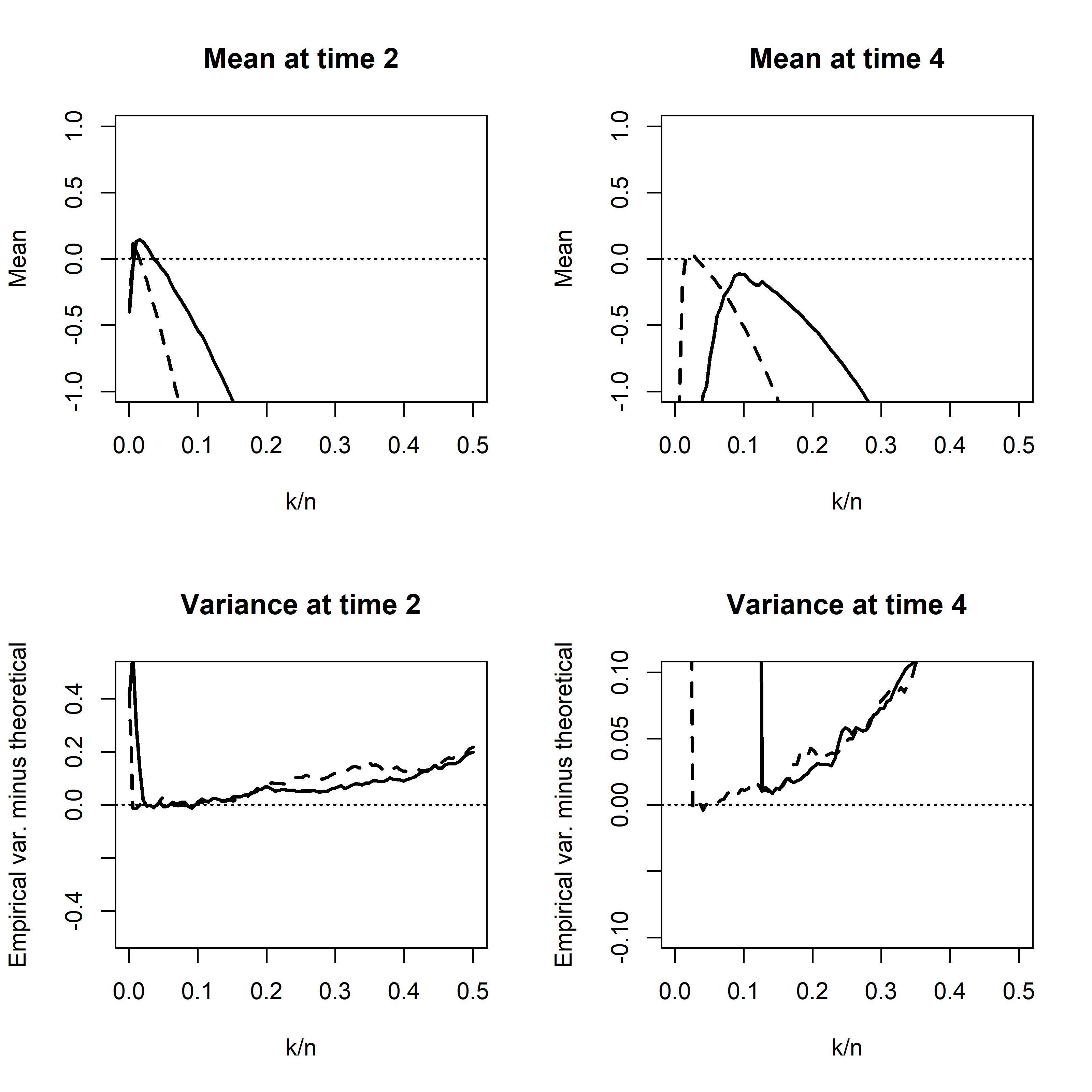}
    \caption{Finite sample behavior of $\sqrt{k}(\mathbb{\Lambda}_{k,n}-\Lambda^\circ)$ (four left panels) and $\sqrt{k}(\F_{k,n}-F^\circ)$ (four right panels). The sample sizes are $n=10,000$ (dashed line) and $n=1,000$ (solid line). The distribution is Burr with $\gamma_X=0.5$ and $\gamma_Y=1.5$, and we have $\mbox{Var}(\mathbb{Z}\circ T(2))=3.7$, $\mbox{Var}(\mathbb{Z}\circ T(4))=24.8$, $(1-F^\circ)^2(2)\mbox{Var}(\mathbb{Z}\circ T(2))=0.2$, $(1-F^\circ)^2(4)\mbox{Var}(\mathbb{Z}\circ T(4))=0.1$.} 
    \label{fig:Ver-burr}
\end{figure}

\begin{figure}[htbp!]
    \centering
            \includegraphics[width=0.49\textwidth]{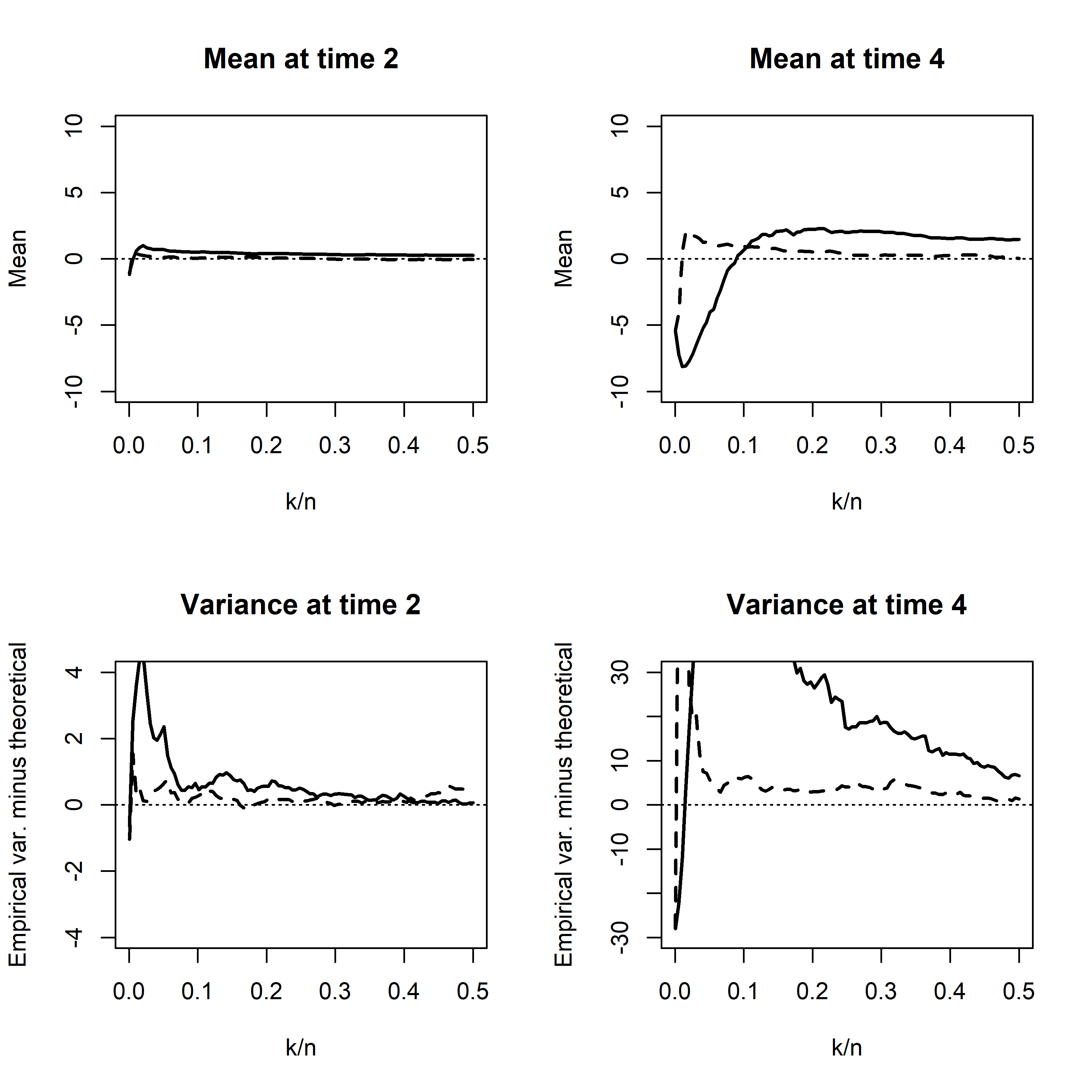}
        \includegraphics[width=0.49\textwidth]{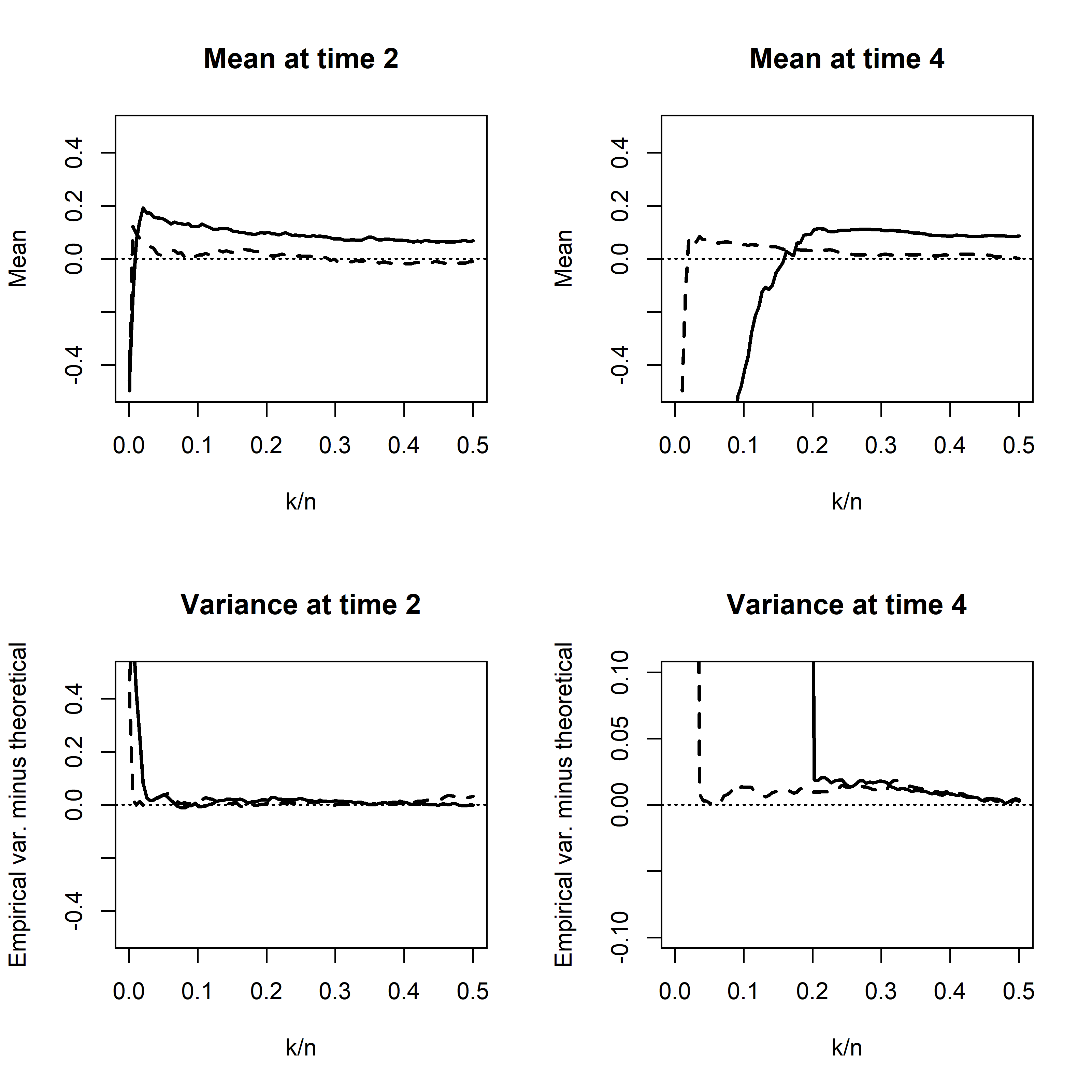}
    \caption{Finite sample behavior of $\sqrt{k}(\mathbb{\Lambda}_{k,n}-\Lambda^\circ)$ (four left panels) and $\sqrt{k}(\F_{k,n}-F^\circ)$ (four right panels). The sample sizes are $n=10,000$ (dashed line) and $n=1,000$ (solid line). The distribution is exact Pareto with $\gamma_X=0.5$ and $\gamma_Y=1.5$, and we have $\mbox{Var}(\mathbb{Z}\circ T(2))=3.7$, $\mbox{Var}(\mathbb{Z}\circ T(4))=24.8$, $(1-F^\circ)^2(2)\mbox{Var}(\mathbb{Z}\circ T(2))=0.2$, $(1-F^\circ)^2(4)\mbox{Var}(\mathbb{Z}\circ T(4))=0.1$.} 
    \label{fig:Ver-par}
\end{figure}

\section{Applications}
\label{sec: Applied}

In this section, we consider applications of the pathwise convergence of the EKM estimator. First, we show how it can be used to recover the consistency and normality of the censored Hill estimator. These asymptotics are already known, but we nonetheless provide a proof which shows how our pathwise identities simplify arguments significantly, compared to \cite{EKM}, and also under different assumptions. Secondly, we use the pathwise convergence to establish convergence of the Extreme Kolmogorv-Smirnov test and Extreme Cram\'er--von Mises statistics, inspired by their non-extreme counterparts. The latter is only possible thanks to our pathwise identities from Theorem \ref{thm:normality}. Thirdly, we show through simulations how these statistics can be used to measure the quality of asymptotic approximations and thus serve as a way to select an appropriate sample fraction $k$. Finally, we apply our results to a French non-life insurance dataset.

\subsection{Tail index estimation}
Recall that $\int_1^\infty \log(t)\dd F_X^\circ(t)=\gamma_X$ and $\amsmathbb{F}_{k,n} \prob F^\circ$, so that the following definition is appropriate as an estimator for $\gamma_X$:  
\begin{definition}[Censored Hill estimator]
\label{def:cen-Hill}
The censored Hill estimator is given by
\begin{align*}
\gamma_{n}=\int_1^\infty \log(t)\dd \amsmathbb{F}_{k,n}(t). 
\end{align*}
\end{definition}
The estimator has been indirectly investigated in \cite{EKM} through EKM integrals, but with U-statistics techniques and under different assumptions. More precisely, the paper investigates estimators of the type
\begin{align}
    \int \varphi \dd \F_{k,n},
    \label{eq:EKMI}
\end{align}
for suitable functions $\varphi$ satisfying some integrability conditions. Here, we show that our pathwise setup can lead to a novel path for the case $\varphi=\log$. Moreover, the same technique easily extends to more general $\varphi$ functions using our results and the functional delta method, though we leave the details to the reader and focus only on $\varphi=\log$.

In fact, consistency and normality follow nearly directly from Theorems \ref{thm:consistency} and \ref{thm:normality}. Indeed, by the latter results, we have that $\int_1^A \log(t)\dd \amsmathbb{F}_{k,n}$ behaves well for any fixed constant $A>1$. The missing piece is to investigate $\int_A^\infty \log(t)\dd \amsmathbb{F}_{k,n}$ for some large $A$, and more particularly to show its asymptotic negligibility. In other words, if for $A$ large enough, the latter integral does not influence the former, then the result follows. Note that we do not need to impose additional assumptions to handle the latter integral in the weak consistency case.  
\begin{theorem}[Weak consistency of the censored Hill estimator]
We have that $$\gamma_{n}\prob\gamma_X.$$
\end{theorem}
\begin{proof}
By consistency on compacts of the EKM function, for any $\eta>1$,
\begin{align*}
\int_1^A \log(t)\dd \amsmathbb{F}_{k,n}(t)\xrightarrow{\pr} \int_1^A \log(t)\dd F^\circ(t)=\gamma_X(1-A^{-1/\gamma_X}).
\end{align*}
Thus the proof is complete if we can show that for every $\eta>0$,
\begin{align}
\lim_{A\to\infty}\limsup_{n\to\infty}\pr\left(\int_{A}^\infty \log(t)\dd\amsmathbb{F}_{k,n}(t)>\eta\right)=0. \label{eq:complete}
\end{align}
Construct the following: let $\{X_i^*\}_{i=1}^k$ and $\{Y_i^*\}_{i=1}^k$ have cdf $F^t$ and $G^t$ respectively. Then let $V_i^*=\min\{X_i^*,Y_i^*\}$ and $\delta_i^*=I(X_i^*\leq Y_i^*)$. Then we have that $\{Z_{n-k,n}/Z_{n-k,n}\}_{i=1}^k$ given $Z_{n-k,n}=t$ has the same distribution as $\{V^*_i\}_{i=1}^k$. Given $Z_{n-k,n}=t$, \eqref{eq:EKM} is then almost surely given by
\begin{align*}
    \amsmathbb{F}^t_{k,n}(x)=1-\prod_{i=1}^k \left[ 1-\frac{\delta^*_{[i,k]}}{k-i+1}\right]^{I(V^*_i\leq x)}. 
\end{align*} 
Hence we use Theorem 3.1 from \cite{Mauro} to get
\begin{align*}
    &\E\left[\int_{A}^{\infty} \log(s)\dd \amsmathbb{F}_{k,n}(s)\Bigg| Z_{n-k,n}=t\right]
    \leq \E\left[1\{X^*_1>A\}\log(X^*_1)\right].
\end{align*}Denote $\varphi(x,A)=I\{x>A\}\log(x)$. Note that $X^*_1$ has the same distribution as $\frac{U_F(t\xi)}{U_F(t)}$, where $\xi$ is Pareto(1). Hence we can use Potter's bounds to establish that there exists a $t_0\geq 1$ such that for $t>t_0$, we have
\begin{align*}
    \varphi\left(X_1^*,A\right)\leq \varphi((1+\epsilon)\xi^{\gamma_X+\epsilon},A),
\end{align*}
and
\begin{align*}
    \E\left[\varphi(X_1^*,A)\right]=\E\left[\varphi\left(\frac{X}{t},A\right)\Bigg| X>t\right]=\frac{\E\left[\varphi\left(\frac{X}{t},A\right)1_{X>t}\right]}{\pr(X>t)} \leq  \frac{\E[\varphi(X,A)]}{\pr(X>t)}:=g(t)
\end{align*}
for all $A>1$. Since $\E[\varphi(X,A)]<\infty$ we have
\begin{align*}
    \E\left[b^{Z_{n-k,n}}(A)\right]&=\E\left[b^{Z_{n-k,n}}(A)I(Z_{n-k,n}>T_0)+b^{Z_{n-k,n}}(A)I(Z_{n-k,n}>T_0)\right]\ \\
    &\leq \E[\varphi((1+\epsilon)\xi^{\gamma_X+\epsilon},A)]\pr(Z_{n-k,n}>T_0)+g(T_0)\pr(Z_{n-k,n}\leq T_0).
\end{align*}
As $n$ goes to infinity the second term goes to 0, since $Z_{n-k,n}\xrightarrow{\pr}\infty$. Then as $A$ goes to $\infty$ the first term goes to 0. Hence \eqref{eq:complete} follows by Markov's inequality. 
\end{proof}
To ensure that $\gamma_n$ is asymptotically normal, additional assumptions are needed. As in Theorem \ref{thm:normality}, we require that the growth of $k$ depends on $B_n$, however, we need to ensure that $\int_A^\infty \log(t)\dd \amsmathbb{F}_{k,n}$ behaves in a tractable way. Thus, assume that $F_X$ satisfies the Von Mises' convergence, i.e. $F_X$ satisfies $$\lim_{t\to \infty} \frac{-F_X'(t)}{t(1-F_X(t))}=\frac{1}{\gamma_X}.$$ Moreover, we further constraint the growth $k$ by imposing a deterministic growth condition depending on $F^{u_n}_X(t)$. 
\begin{theorem}[Weak convergence of the censored Hill estimator]
Let $k=k_n$ grow slow enough such that for some $s_0\in(0,1)$, $\sqrt{k}B_n(s_0)\to0$, and let $p>1/2$. Assume $F_X$ satisfies the Von Mises' convergence and 
\begin{align}
    \lim_{n\to \infty}\sqrt{k}(F^{u_n}_X(t)-F_X^\circ(t))= 0. \label{ass:un}
\end{align}
Then $$\sqrt{k}(\gamma_{n}-\gamma_X)\dist\mathcal{N}\left(0,\gamma_X^2\frac{p}{2p-1}\right).$$
\end{theorem}
\begin{proof}
For $A>1$, we have that
\begin{align*}
    \int_{1}^A \frac{\sqrt{k}({\amsmathbb{F}}_{k,n}(t)-F_X^\circ(t))}{t}\dd t \dist \int_1^A\frac{1-F_X^\circ(t)}{t}\, \amsmathbb{Z}\circ T(t) \dd t
\end{align*}
by Theorem \ref{thm:normality}. If
for every $\eta>0$, we have
\begin{align}
\lim_{A\to\infty}\limsup_{n\to\infty}\pr\left(\int_{A}^\infty \sqrt{k}\frac{\amsmathbb{F}_{k,n}-F_X^\circ (t)}{t}\dd t>\eta\right)=0. \label{eq:normality}
\end{align}
then 
\begin{align*}
    \sqrt{k}(\gamma_n-\gamma) \dist \int_1^\infty \frac{1-F_X^{\circ}}{t}\amsmathbb{Z}\circ T(t)\dd t.
\end{align*}
 We have
\begin{align}
    \int_{A}^\infty \sqrt{k}\frac{\amsmathbb{F}_{k,n}-F_X^\circ (t)}{t}\dd t&= \int_{A}^\infty \sqrt{k}\frac{\amsmathbb{F}_{k,n}-F_X^{Z_{n.k,n}} (t)}{t}\dd t\label{eq:1} \\&\quad+ \int_{A}^\infty \sqrt{k}\frac{F_X^{Z_{n-k,n}}(t)-F_X^{u_n} (t)}{t}\dd t\label{eq:2}\\
    &\quad +\int_{A}^\infty \sqrt{k}\frac{F_X^{u_n}(t)-F_X^\circ(t)}{t}\dd t, \label{eq:3}
\end{align}
where $F_X^t(x)=\frac{1-F_X(tx)}{1-F_X(t)}$. We have that
\eqref{eq:1} goes to 0 by p.xiv in the Supplementary Material of \cite{EKM}, and \eqref{eq:3} goes to 0 by Assumption \ref{ass:un}. Now we turn our attention to \eqref{eq:2}. Denote $\Bar{F}_X(t)=1-F_X(t)$. We have that 
\begin{align*}
    \Bar{F}_X^{u_n}(tZ_{n-k,n}/u_n)&=\frac{1-F(Z_{n-k,n}t)}{1-F(u_n)} = \Bar{F}_X^{Z_{n-k,n}}(t)\frac{\Bar{F}_X(Z_{n-k,n})}{\Bar{F}_X(u_n)}.
\end{align*}
So \eqref{eq:2} can be written as
\begin{align}
    &\int_{A}^\infty \sqrt{k}\frac{\Bar{F}_X^{Z_{n-k,n}}(t)-\Bar{F}_X^{u_n} (t)}{t}\dd t=\nonumber\\ &\quad \quad \quad \quad 
    \sqrt{k}\left(\frac{\Bar{F}(u_n)}{\Bar{F}(Z_{n-k,n})}-1\right)\int_A^\infty \frac{\Bar{F}_X^{u_n}\left(t\frac{Z_{n,k,n}}{u_n}\right)}{t}\dd t \label{eq:5}\\
    &\quad\quad \quad \quad +\int_A^{\infty} \sqrt{k} \frac{\Bar{F}_X^{u_n}\left(t\frac{Z_{n,k,n}}{u_n}\right)-\Bar{F}_X^{u_n} (t)}{t}\dd t. \label{eq:4}  
\end{align}
We first look at \eqref{eq:4}. By a change of variables, we may rewrite it as
\begin{align*}
   \int_A^{\infty} \sqrt{k} \frac{\Bar{F}_X^{u_n}(t)-\Bar{F}_X^{u_n} (t)}{t}\dd t + \int_{A\frac{Z_{n,k,n}}{u_n}}^A \frac{\Bar{F}_X^{u_n} (t)}{t}\dd t=\int_{A\frac{Z_{n,k,n}}{u_n}}^A \sqrt{k}\frac{\Bar{F}_X^{u_n} (t)}{t}\dd t.
\end{align*}
Note that $\frac{Z_{n,k,n}}{u_n}\prob 1$ by Theorem \ref{thm:prob} as $n/k \to \infty$, hence for $\epsilon>0$, and $n/k$ large enough, we have with arbitrary high probability that
\begin{align*}
    \int_{A\frac{Z_{n,k,n}}{u_n}}^A \sqrt{k}\frac{\Bar{F}_X^{u_n} (t)}{t}\dd t&\le \left(\frac{Z_{n,k,n}}{u_n}A-A\right)\sqrt{k}\sup_{t\in (A(1-\epsilon),A(1+\epsilon))}\frac{\Bar{F}_X^{u_n}(t)}{t}\\
    &= \left(\frac{Z_{n,k,n}}{u_n}A-A\right)\sqrt{k}\frac{\Bar{F}_X^{u_n}((1-\epsilon)A)}{(1-\epsilon)A} \\
    &= \sqrt{k}\left(\frac{Z_{n,k,n}}{u_n}-1\right) C \Bar{F}_X^{u_n}((1-\epsilon)A).
\end{align*}
By Theorem \ref{thm:quan}, we have that $\sqrt{k}\left(\frac{Z_{n,k,n}}{u_n}-1\right)\dist \mathcal{N}(0,\gamma)$. Since $\Bar{F}_X^{u_n}$ goes to $\Bar{F}_X^\circ$, we have that the integral vanishes in probability as $n/k\to \infty$ and thereafter $A\to \infty$. Now look at \eqref{eq:5}. By Potter bounds, the integral vanishes as $A\to\infty$, so the result follows if $\sqrt{k}\left(\frac{\Bar{F}(u_n)}{\Bar{F}(Z_{n-k,n})}-1\right)$ is bounded in probability. By Theorem \ref{thm:quan}, we have that
\begin{align*}
    \frac{\sqrt{k}}{u_n}\left(Z_{n,k,n}-u_n\right)\dist \mathcal{N}(0,\gamma^2).
\end{align*}
So using the delta method with $\Bar{F}$, we have
\begin{align*}
    D_n:=\sqrt{k}\frac{\Bar{F}(u_n)}{-u_n F'(u_n)}\left(\frac{\Bar{F}(Z_{n-k,n})}{\Bar{F}(u_n)}-1\right) \dist \mathcal{N}(0,\gamma^2).
\end{align*}
Then notice that
\begin{align*}
    \sqrt{k}\left(\frac{\Bar{F}(Z_{n-k,n})}{\Bar{F}(u_n)}-1\right)=\frac{-u_nF'(u_n)}{\Bar{F}(u_n)}D_n,
\end{align*}
where the fraction on the right-hand side converges to a constant by the Von Mises convergence. This establishes that $\sqrt{k}\left(\frac{\Bar{F}(Z_{n-k,n})}{\Bar{F}(u_n)}-1\right)$ is bounded in probability.

The variance of the Gaussian limit is calculated from
\begin{align*}
\int_1^\infty\frac{1-F_X^\circ(t)}{t}\, \amsmathbb{Z}\circ T(t) \dd t=\int_1^\infty t^{-1/\gamma_X-1}\amsmathbb{Z}(t^{-1/\gamma})\dd t 
=\gamma\int_0^1 t^{p-1} \amsmathbb{Z}(t)\dd t,
\end{align*}
and since $p=\gamma/\gamma_X$ then the variance is given by Lemma \ref{lemma:computation} by
\begin{align*}
\gamma^2\int_{0}^1\int_{0}^1t^{p-1} s^{p-1} \{s^{-1}\wedge t^{-1}-1\} \dd s\dd t=\gamma^2\frac{1}{p(2p-1)} =\gamma^2_X \frac{p}{(2p-1)}.
\end{align*}
\end{proof}

\subsection{Goodness-of-fit statistics}
In addition to estimating the tail index, the EKM integral in \eqref{eq:EKMI} can estimate several other relevant tail quantities by selecting a suitable $\varphi$ -- for example the conditional tail moments. The common feature of these estimators is that they depend on how well $\F_{k,n}$ estimates $F_X^\circ$. Here the choice of $k$ is of great importance and an open problem. Generally, selecting a low $k$ will result in a small bias but a large variance, with the opposite case for a large $k$. Balancing these two effects is an intricate matter. A surprisingly simple rule of thumb with good behavior for very small samples is often to select $k=0.2n$. As $n$ grows, more sophisticated rules can be formulated. Therefore, in this section we present an adaptive selection rule based on the observed sample. More precisely, our rule is based on GoF statistics, which are available since we have shown pathwise convergence.
 Specifically, we use the Extreme Kolmogorov--Smirnov and Extreme Cram\'er--von Mises statistics, which are respectively defined by
\begin{align*}
\sqrt{k}||\F_{k,n}-F||_\infty \quad \text{and} \quad k\int(\F_{k,n}-F_X^\circ)^2\dd F_X^\circ.
\end{align*}
By Theorem \ref{thm:normality}, the distributions of the statistics are known.
\begin{corollary}[Asymptotic representation of goodness-of-fit statistics]
The asymptotic representations of the Extreme Kolmogorov--Smirnov and Extreme Cram\'er--von Mises statistics are given, respectively, by 
\begin{align} \label{eq:GoF}
    ||(1-F^\circ_X)\, \amsmathbb{Z}\circ T||_\infty \quad \text{and} \quad  \int_1^\infty ((1-F^\circ_X)\, \amsmathbb{Z}\circ T)^2(s)\dd F^\circ_X(s).
\end{align}
\end{corollary}
\begin{proof} 
 Since the maps
\begin{align*}
z\mapsto||z||_\infty,\quad z\to\int z^2(t)\dd t, 
\end{align*}
from $D([-\infty,\infty])$ to $\amsmathbb{R}$ are continuous in the sup-norm, the result follows from the continuous mapping theorem and Theorem \ref{thm:normality}.
\end{proof}

Note that the variance of $\amsmathbb{Z} \circ T(s)$ goes to $\infty$ as $s$ goes to $\infty$, hence the ENA-version of the statistics cannot be defined directly. It is also interesting to observe that the above representations depend on $T$, and thus are not distribution-free, as is the case for their non-extreme counterparts.

Since $F_X^\circ$ is unknown we use in practice the following plug-in approximations
\begin{align}
    \sqrt{k}||\F_{k,n}-F^\circ_{X,\hat{\gamma}_n}||_\infty \quad \text{and} \quad  k\int(\F_{k,n}-F^\circ_{X,\hat{\gamma}_n})^2\dd F^\circ_{X,\hat{\gamma}_n}, \label{eq:emp-GoF}  
\end{align}
where $F^\circ_{X,\hat{\gamma}}(s)=1-s^{1/\hat{\gamma}}$ and $\hat{\gamma}$ is an estimate of $\gamma_X$. In other words, for our proposed selection rules, we use \eqref{eq:emp-GoF} as a proxy for \eqref{eq:GoF}. 

In Figure \ref{fig:all} we show the pathwise convergence for the Kolmogorov--Smirnov and Cram\'er--von Mises statistics for different distributions in the Fr\'echet-domain. As expected, we observe that the paths from the Fr\'echet and Burr distributions converge to the Pareto path as $k$ decreases. Importantly, they converge from above. This provides inspiration to formulate a rule by selecting an upper bound. 
 \begin{figure}[htbp!]
    \centering
        \includegraphics[width=0.4\textwidth]{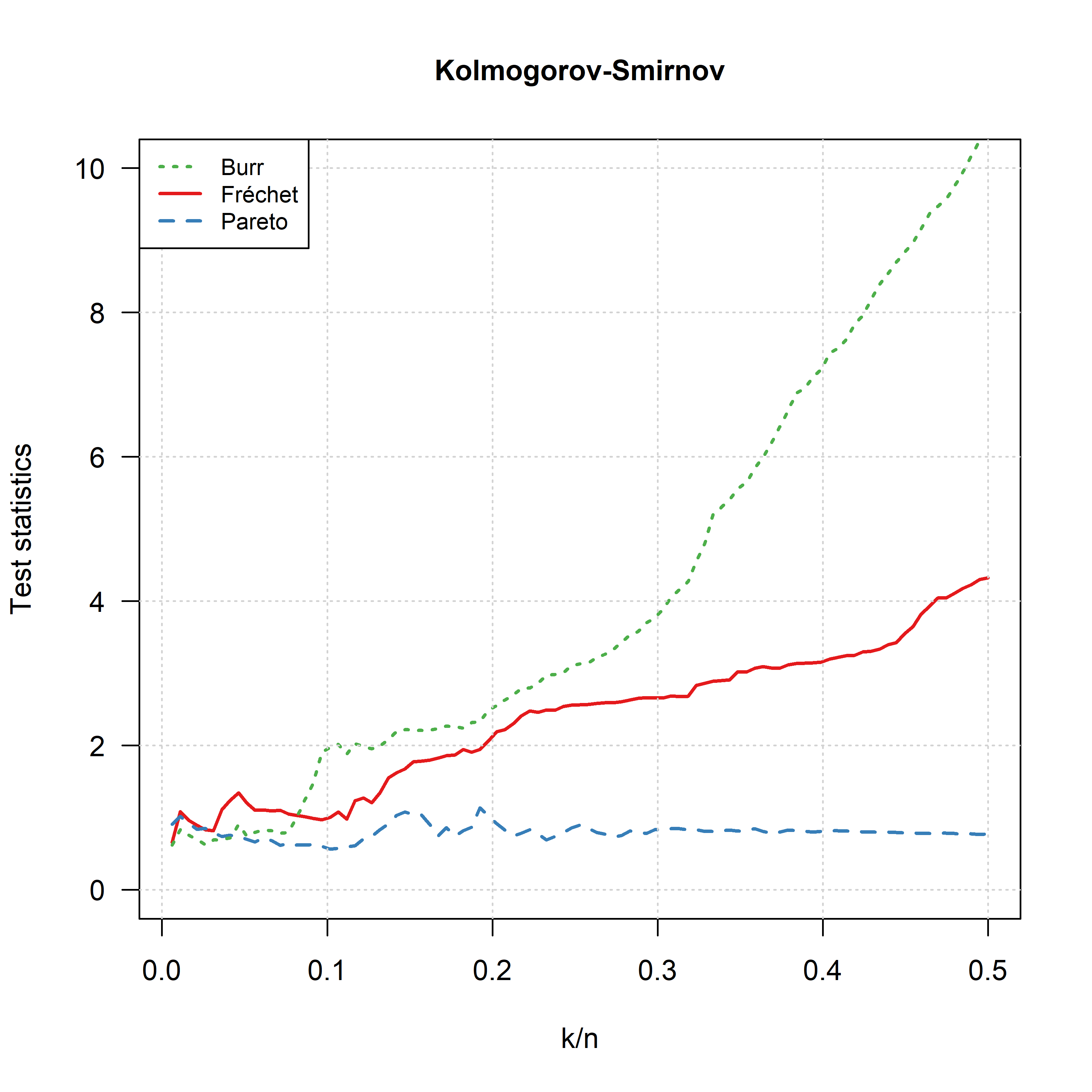}
        \includegraphics[width=0.4\textwidth]{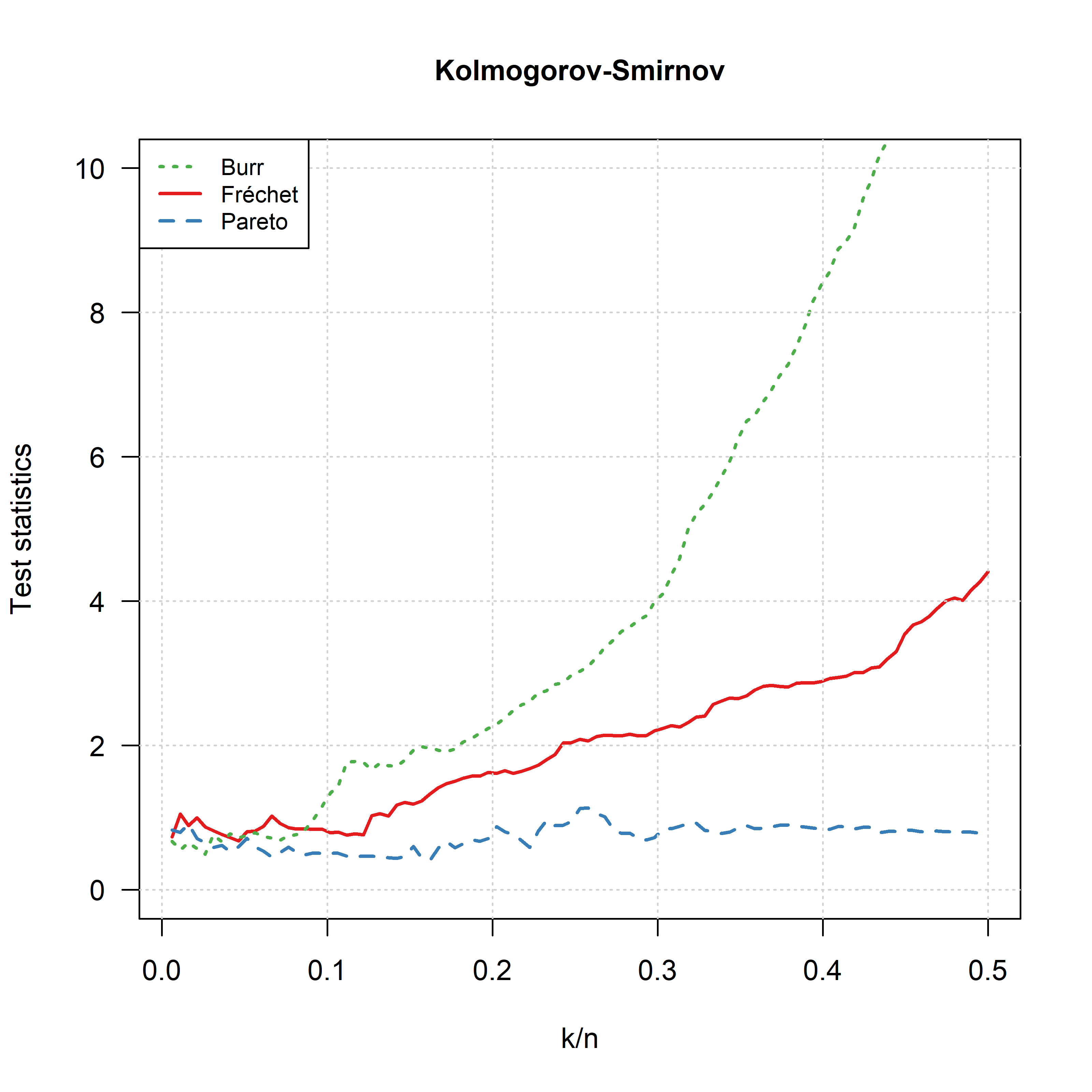}
        \includegraphics[width=0.4\textwidth]{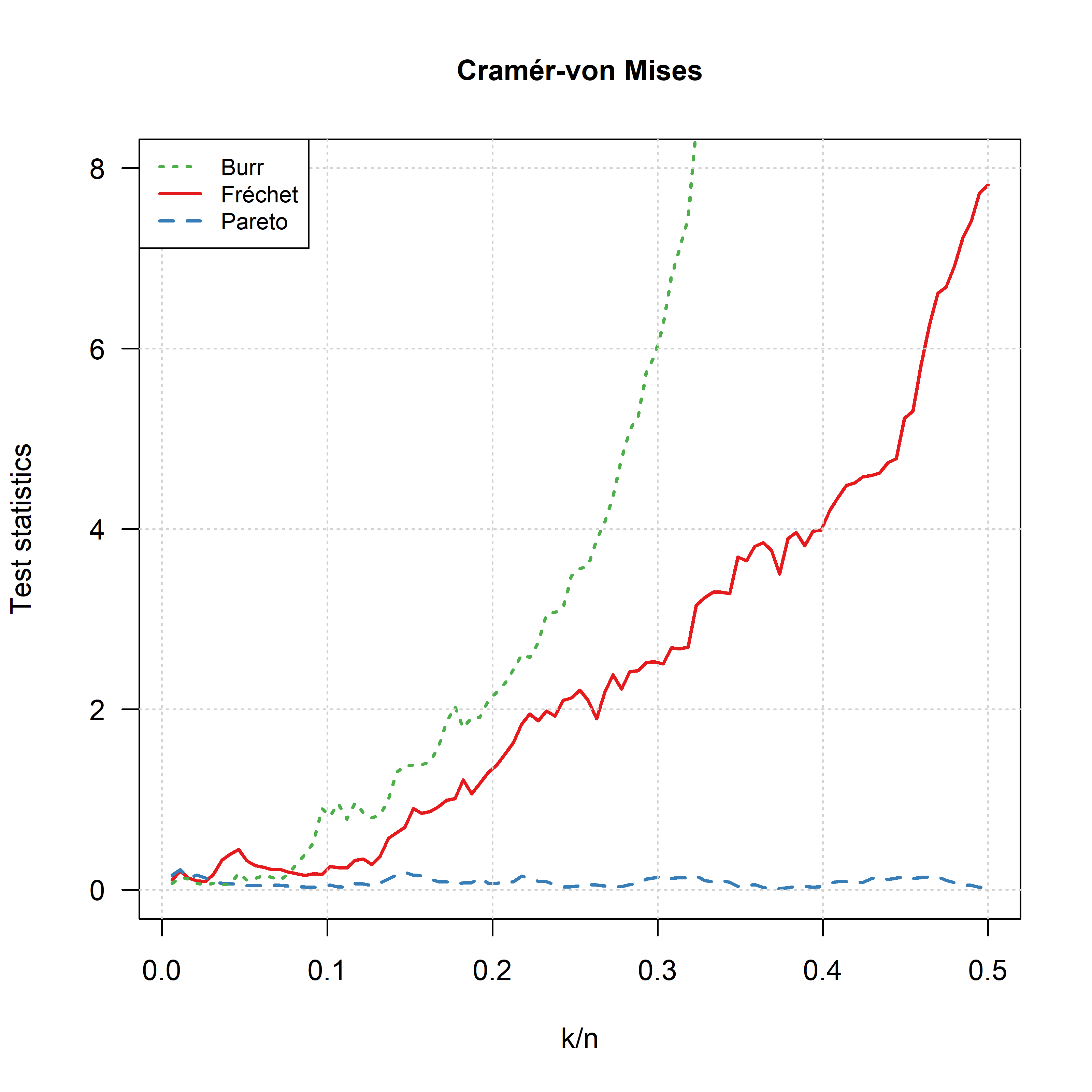}
        \includegraphics[width=0.4\textwidth]{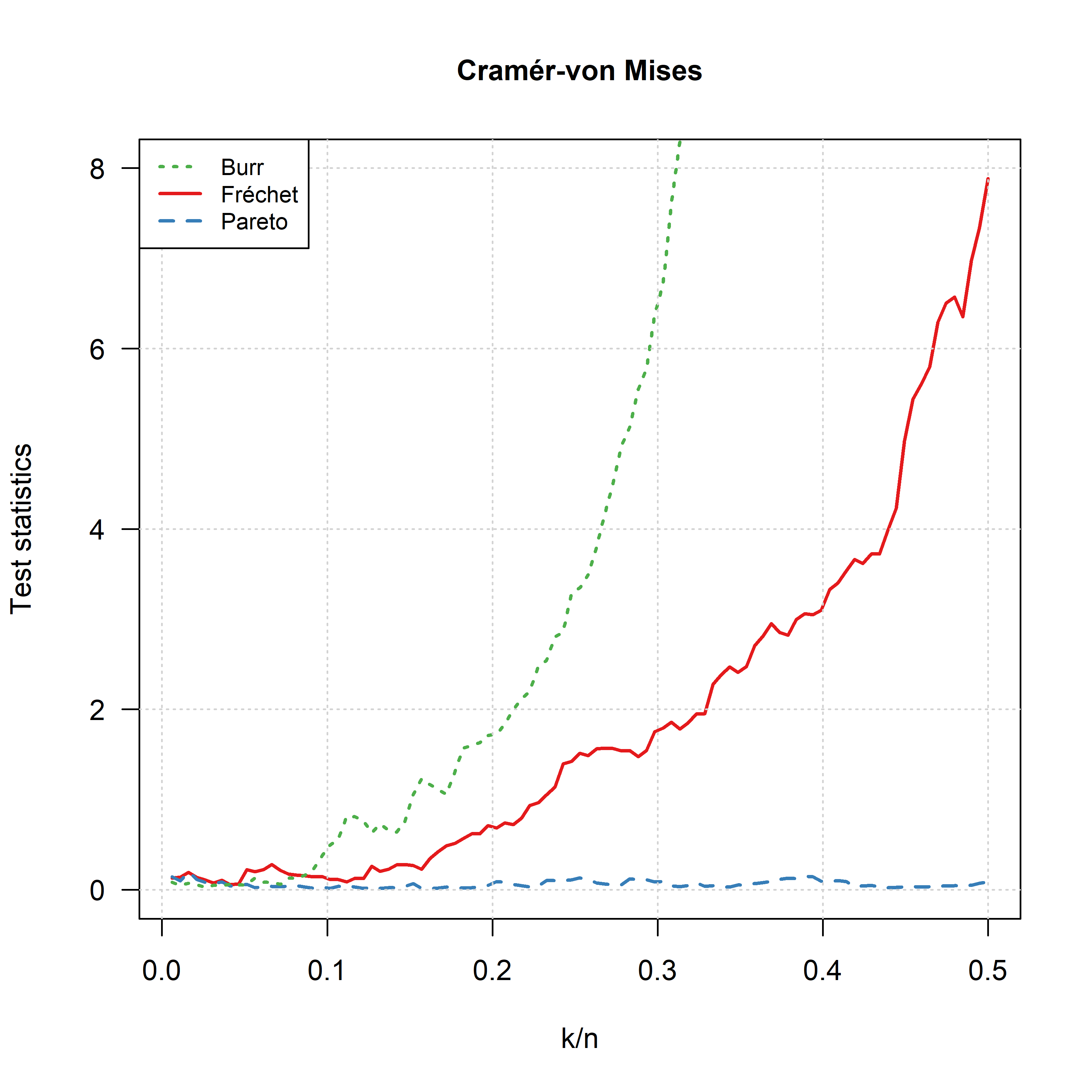}
        \includegraphics[width=0.4\textwidth]{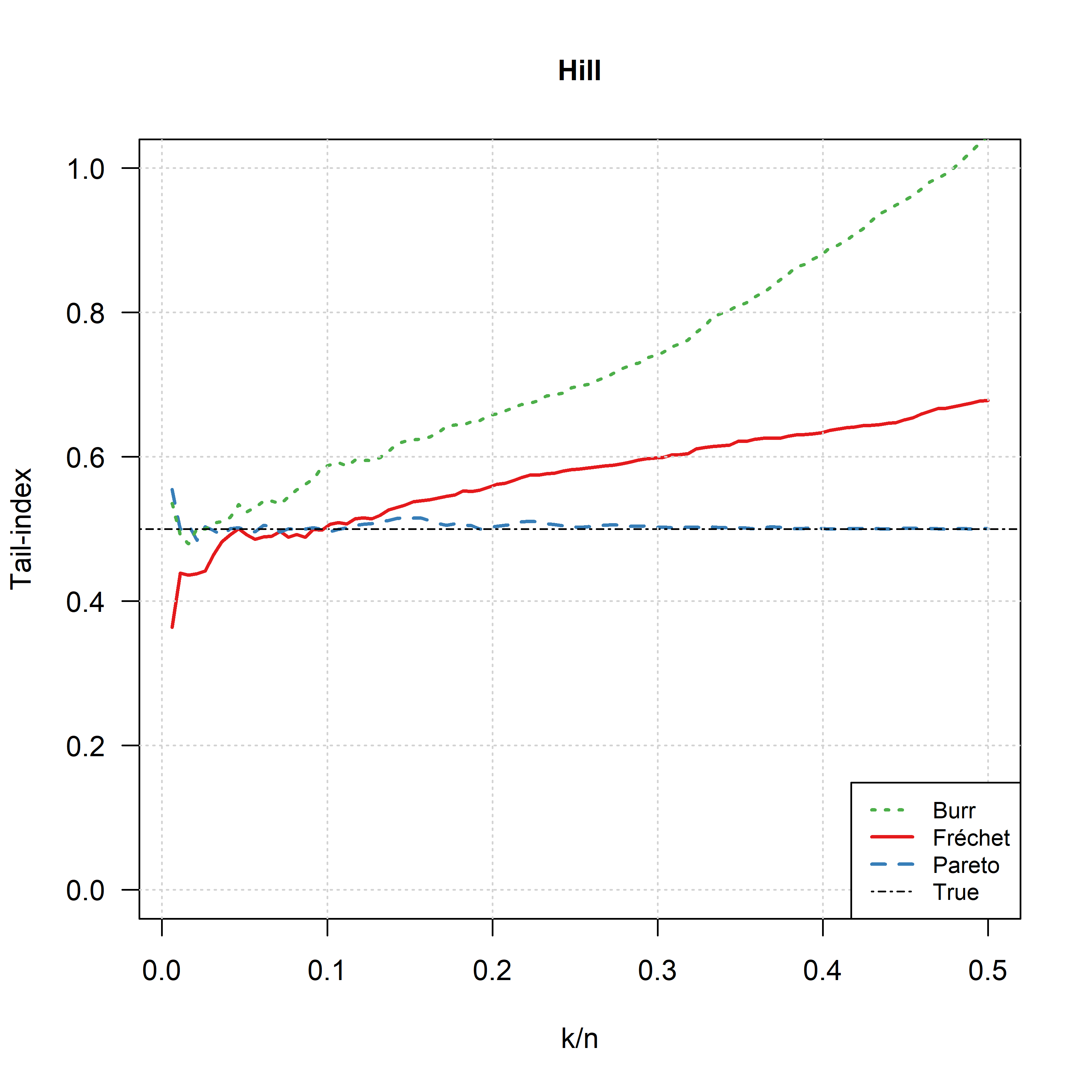}
        \includegraphics[width=0.4\textwidth]{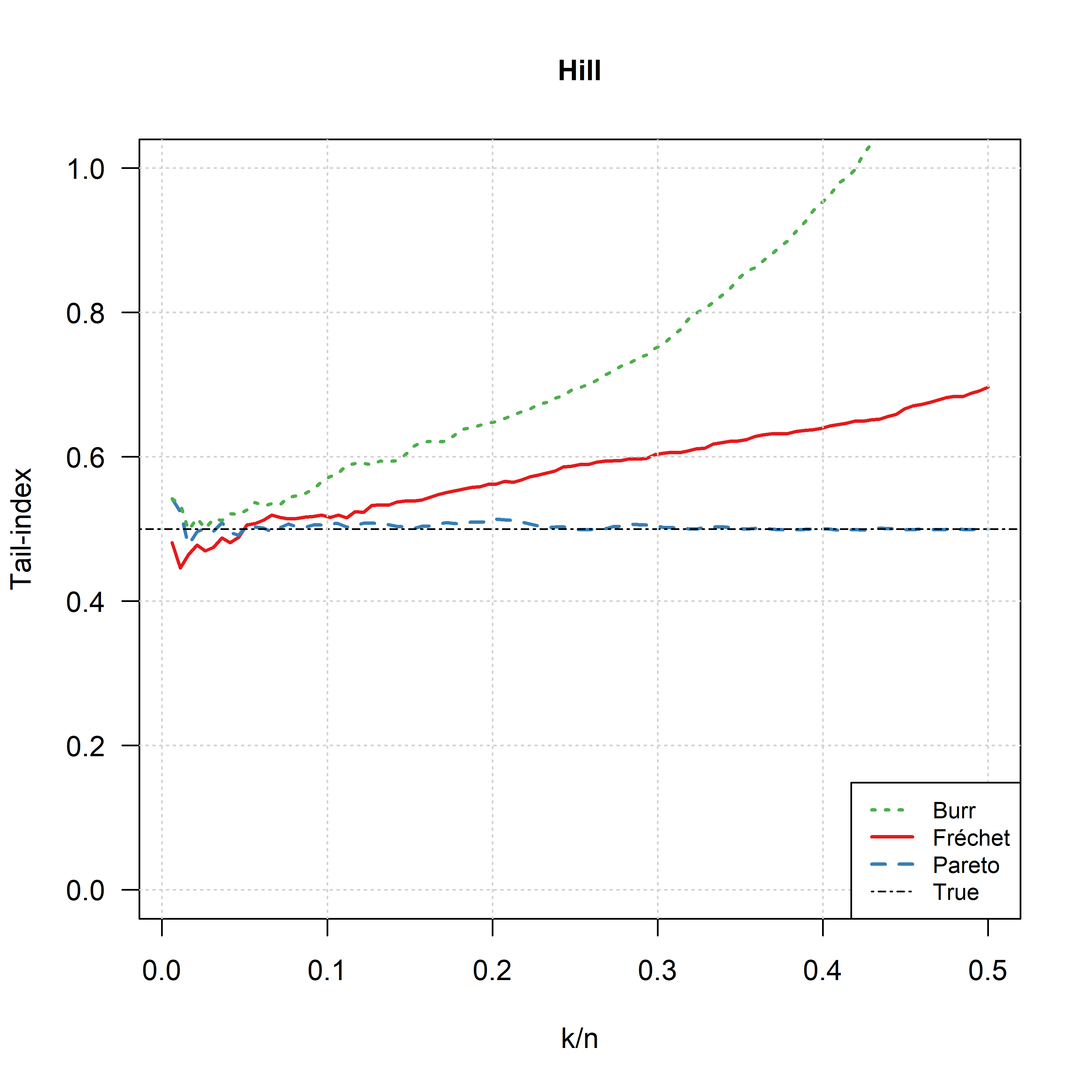}
    \caption{The evolution of the Extreme Kolmogorov--Smirnov (top) and Extreme Cram\'er--von Mises statistics (middle), together with the censored Hill (bottom). The left panels show the $\gamma_X=0.5$, $\gamma_Y=1.5$ case, while the right panes show the $\gamma_X=0.5$, $\gamma_Y=5$ case. The sample size is $n=10,000$.} 
    \label{fig:all}
\end{figure}

We propose two different selection rules:
\begin{itemize}
    \item[i)] Pick the largest $k$ such that $\sqrt{k}||\F_{k,n}-F^\circ_{X,\hat{\gamma}_n}||_\infty< L$ for some $L>0$.
    \item[ii)] Pick the largest $k$ such that $k\int(\F_{k,n}-F^\circ_{X,\hat{\gamma}_n})^2\dd F^\circ_{X,\hat{\gamma}_n} $ for some $L>0$.
\end{itemize}
The choice of $L$ is of course pivotal. If we pick $L$ too small, then the GoF statistics can never get below $L$, and likewise, we want to avoid an overly large $L$, since then the GoF statistics would lie below it for all $k$. Both instances would lead to an inadequate procedure for $k$ selection. We want to hit a sweet spot, where the $p$-value captures the variability in the limit distribution but still can be used to differentiate between when the discrepancy of $\F_{k,n}$ and $F_X^\circ$ is too great. We see in the next subsection that despite this caveat, the above procedures are robust to changes in $L$, so the most important aspect is simply avoiding degeneracy.

\subsection{Simulation study}
\label{sec:sim}
To test the selection rule from the previous section, we simulate from the following distributions:

\begin{itemize}
    \item Burr with density $f(x)=\frac{\alpha\tau x^\tau}{x(1+x^\tau)^\alpha}$, where $\tau, \alpha>0$. Throughout the section $\alpha=1$, where $\tau$ varies. In this case the corresponding tail index is $\gamma=\frac{1}{\tau}$.
    \item Fr\'echet with density $f(x)=\alpha\left(x\right)^{-1-\alpha}\exp(-x^{-\alpha})$, where $\alpha,s>0$. In this case the tail index is giving by $\gamma=\frac{1}{\alpha}$.
\end{itemize}
In the simulation, $X$ and $Y$ belong to the same family of distributions, but they have different tail indices, respectively given by $\gamma_X$ and $\gamma_Y$. We keep $\gamma_X=0.5$ fixed in all simulations, where either $\gamma_Y=0.8$ or $\gamma_Y=1.5$ to investigate the behavior of the rules under different censoring schemes.   

Thus, we compare the following three threshold selection rules:
\begin{itemize}
    \item[S1:] Pick $k=0.2n$.
    \item[S2:] Pick the largest $k$ such that $\sqrt{k}||\F_{k,n}-F^\circ_{X,\hat{\gamma}_n}||_\infty< L$ for some $L>0$.
    \item[S3:] Pick the largest $k$ such that $\int_0^1 ((1-F^\circ_{X,\hat{\gamma}_n})\, \amsmathbb{Z}\circ T)^2(s)\dd F^\circ_{X,\hat{\gamma}_n}< L$ for some $L>0$.
\end{itemize}
If there does not exist a $k$ such that the quantities in S2 or S3 are below $L$ for any $k$, then we let $k=0.2n$.

To assess the quality of the selection rules, we the in a second step use the selected $k$ to estimate $\gamma_X$. For the estimator for $\gamma_X$ we use the censored Hill estimator as per Definition \ref{def:cen-Hill}.

\begin{table}[htbp!]
    \centering
    \begin{tabular}{c|ccccccc}
        $n$ &  S1 & S2 & S3 &  S2 & S3 & S2 & S3\\\hline
         & & \multicolumn{2}{|c|}{$L_1$} & \multicolumn{2}{|c}{$L_2$} & \multicolumn{2}{|c|}{$L_3$}\\\hline
        1,000 & 2.3 & 4.7 & 1.8& 6.8 & 2.6 & 9.0 & 3.7\\
        5,000 & 2.0 & 1.2 & 1.1& 1.4 & 1.2 & 1.9 & 1.4\\
        10,000 & 2.5 & 0.8 & 1.1 & 1.0 & 1.0 & 1.4 & 1.0 \\
        50,000 & 2.2 & 1.1 & 1.0& 1.2 & 0.8 & 0.7 & 0.6 
    \end{tabular}
    \caption{Burr distribution: 100$\times$ MSE based on $500$ simulations, for different selection rules. Here, $\gamma_X=0.5$, $\gamma_Y=0.8$.}
    \label{tab:1}
\end{table}

\begin{table}[htbp!]
    \centering
    \begin{tabular}{c|ccccccc}
        $n$ &  S1 & S2& S1 & S2 & S3 &  S2 & S3 \\\hline
         & & \multicolumn{2}{|c|}{$L_1$} & \multicolumn{2}{|c|}{$L_2$} & \multicolumn{2}{|c|}{$L_3$} \\\hline
        1,000 & 0.5 & 1.8 & 1.1& 1.9 & 1.4 & 2.0 & 1.7  \\
        5,000 & 0.3 &  0.8 & 0.6& 0.9 & 0.6 & 1.0 & 0.8 \\
        10,000 & 0.4 &  0.8 & 0.6 & 0.5 & 0.5 & 0.6 & 0.5\\ 
        50,000 & 0.4 &  0.4 & 0.6 & 0.5 & 0.6 & 0.3 & 0.3
        
    \end{tabular}
    \caption{Fr\'echet distribution: 100$\times$ MSE based on $500$ simulations, for different selection rules. Here, $\gamma_X=0.5$, $\gamma_Y=0.8$.}
    \label{tab:KM-frec-p066}
\end{table}

\begin{table}[htbp!]
    \centering
    \begin{tabular}{c|ccccccc}
        $n$ &  S1 & S2 & S3 &  S2 & S3 & S2 & S3\\\hline
         & & \multicolumn{2}{|c|}{$L_1$} & \multicolumn{2}{|c}{$L_2$} & \multicolumn{2}{|c|}{$L_3$}\\\hline
        1,000 & 2.1 & 4.1 & 1.5& 5.9 & 3.6 & 7.9 & 5.2\\
        5,000 & 2.1 & 1.2 & 0.8& 1.7 & 1.1 & 2.3 & 1.5 \\
        10,000 & 2.1 & 0.8 & 0.5 & 1.0 & 0.7 & 1.4 & 0.9 \\
        50,000 & 2.1 & 0.3 & 0.2& 0.3 & 0.2 & 0.4 & 0.3
    \end{tabular}
    \caption{Burr distribution: 100$\times$ MSE based on $500$ simulations, for different selection rules. Here, $\gamma_X=0.5$, $\gamma_Y=1.5$.}
    \label{tab:2}
\end{table}

\begin{table}[htbp!]
    \centering
    \begin{tabular}{c|ccccccc}
        $n$ &  S1 & S2 & S3 &  S2 & S3 & S2 & S3\\\hline
         & & \multicolumn{2}{|c|}{$L_1$} & \multicolumn{2}{|c}{$L_2$} & \multicolumn{2}{|c|}{$L_3$}\\\hline
        1,000 & 0.4 & 1.9 & 0.8& 2.4 & 1.6 & 2.7 & 2.2 \\
        5,000 & 0.3 & 0.6 & 0.4& 0.9 & 0.5 & 1.1 & 0.9 \\
        10,000 & 0.3 & 0.4 & 0.3 & 0.5 & 0.3 & 0.7 & 0.5 \\ 
        50,000 & 0.3 & 0.1 & 0.1& 0.2 & 0.1 & 0.2 & 0.2 
        
    \end{tabular}
    \caption{Fr\'echet distribution: 100$\times$ MSE based on $500$ simulations, for different selection rules. Here, $\gamma_X=0.5$, $\gamma_Y=1.5$.}
    \label{tab:KM-frec-p075}
\end{table}

Tables \ref{tab:1}, \ref{tab:KM-frec-p066}, \ref{tab:2} and \ref{tab:KM-frec-p075} consist of average mean squared error (MSE) based on different choices of $L$ and numbers of observations. For S2 we have $L_1=1.5$, $L_2=1.75$ and $L_3=2$. For S3 we have $L_1=0.25$, $L_2=0.5$ and $L_3=0.75$. The choices of $L$ are inspired by the plots in Figure \ref{fig:all}. We observe, that when the data is simulated from a Burr distribution then S2 and S3 outperform S1 for large $n$. It is only for $n=1,000$ where S1 is best. When the data is Fr\'echet the picture is less clear. For $\gamma_Y=0.8$, i.e. heavy censoring, S1 seems to slightly outperform S2 and S3. It is only when using $L_3$ and $n=50,000$ that S1 is not better. For $\gamma_Y=1.5$, i.e. more light censoring, we see that for $n=50,000$ S2 and S3 are preferred for all values $L$. However, S1 is preferred for smaller $n$. Overall, we can conclude that S2 and S3 seem stable across the different values of $L$, and that S2 and S3 are better than S1 for large sample sizes, especially when the data is Burr distributed, which has the largest bias in our simulation study, and the three strategies seem more equivalent when there is data is Fr\'echet-distributed, which exhibits much lower bias.

\subsection{Non-life insurance data application}

In this section, we apply the two selection strategies to the dataset \texttt{freclaimset3dam9207}, from the \texttt{R} package \texttt{CASdatasets}. It consists of 109,992 French insurance claims between 1992 and 2007, more specifically for a damage guarantee for motor insurance. Since the claims are paid yearly until settlement, we consider an unsettled (or open) observation as a right-censored observation. The dataset does not specify a settlement date, so we proxy it as the year where the yearly payment stop growing. For a more detailed analysis of the data see \cite{EKM-co}.

 A Hill plot together the the evolution of the two test statistics can be seen in Figure \ref{fig:data}. The two estimates are based on $L=0.5$ and $L=1.75$ for respectively Extreme Cram\'er--von Mises (CM) and Extreme Kolmogorov--Smirnov (KS). The estimates based on the two statistics are similar. Precisely, they are respectively given by $0.55$ and $0.59$. This is in contrast to the naive Rule-of-thumb (RoT), which estimates the tail-index at the astonishing $1.01$. In Figure \ref{fig:data-test}, we see the estimated EKM estimator and $T(s)$ based on CM and RoT. We observe that the EKM fits $T(s)$ better based on CM than on RoT. This corroborates that the estimate from CM is more accurate than the one from RoT. The figure doesn't include KS for readability sake, however, the analysis for KS compared to RoT is very similar. 
 
\begin{figure}[htbp!]
    \centering
        \includegraphics[width=0.49\textwidth]{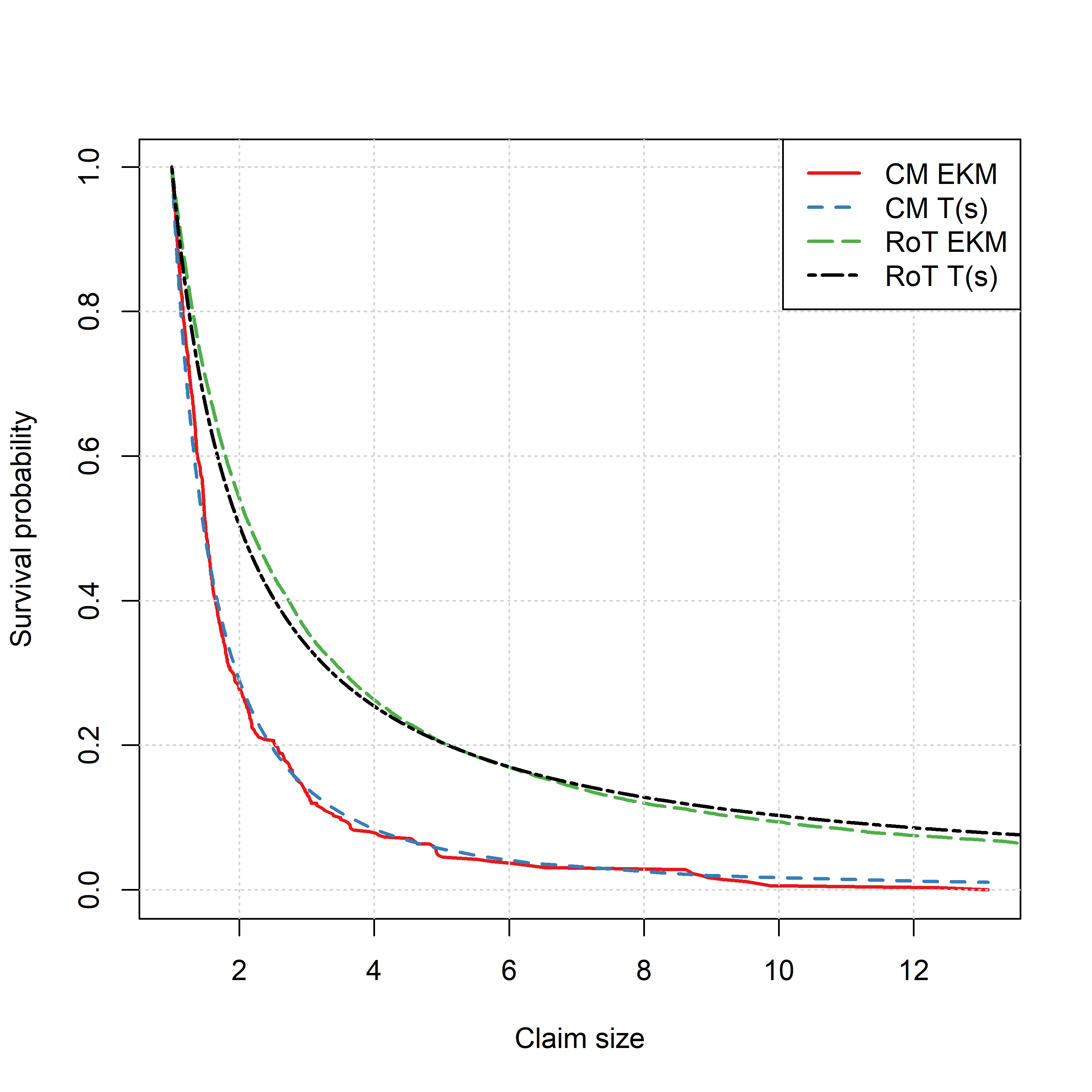}
    \caption{Comparison of the EKM estimator and $T(s)$ based on different selection rules.} 
    \label{fig:data-test}
\end{figure}

In Figure \ref{fig:data-rolling}, we provide automatically selected estimates of the tail index based on a $4$-year rolling window. Specifically, the first tail estimate is calculated using data from 1992 to 1996. The subsequent estimate is based on data from 1993 to 1997, and this pattern continues, with each estimate shifting the window forward by one year. This method allows us to observe how stable our estimates based on the different selection rules are. We can observe that there are a few occasions where the selection rule based on the statistics is the same as the RoT, since the respective test statistics do not get below the threshold $L$ and hence equals RoT. However, the CM and KS mostly  estimate the tail index to be lower than the RoT estimate. As the selection rule based on CM and KS are focused on minimizing bias, we explain the different with that the RoT estimate contain significant positive bias. We can further observe the CM seem to be more stable than KS, although this might be specific to this dataset. 
\begin{figure}[htbp!]
    \centering
        \includegraphics[width=0.49\textwidth]{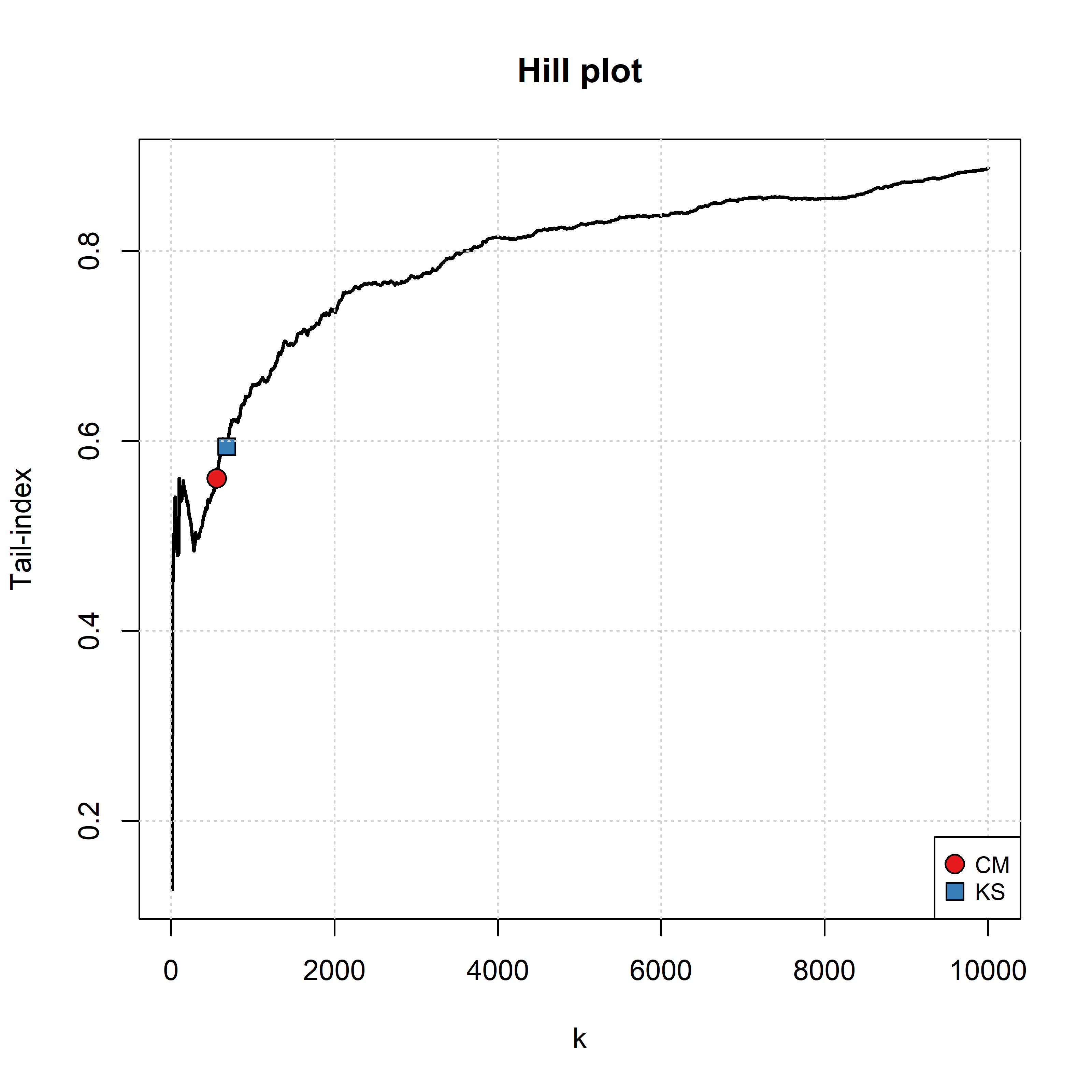}
        \includegraphics[width=0.49\textwidth]{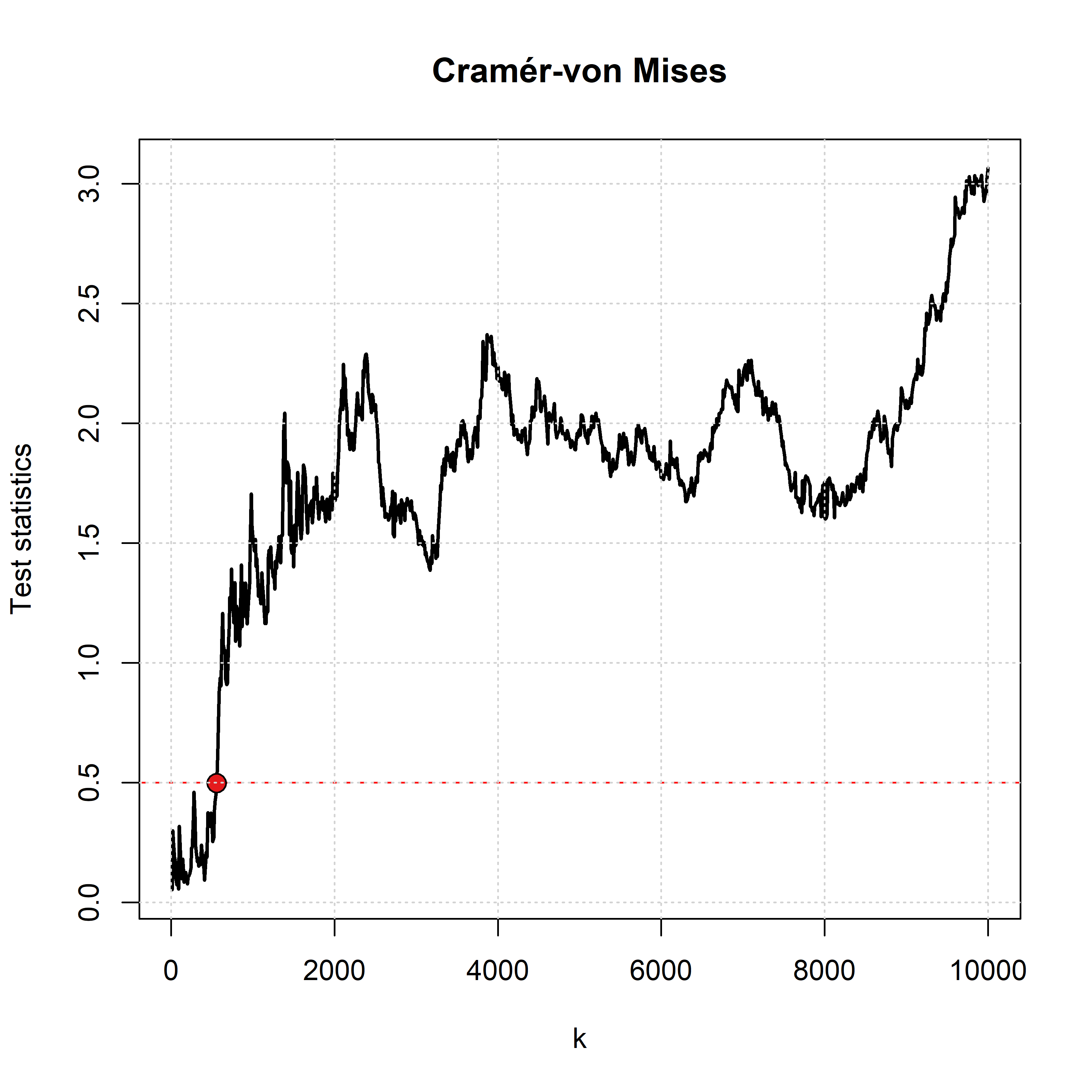}
        \includegraphics[width=0.49\textwidth]{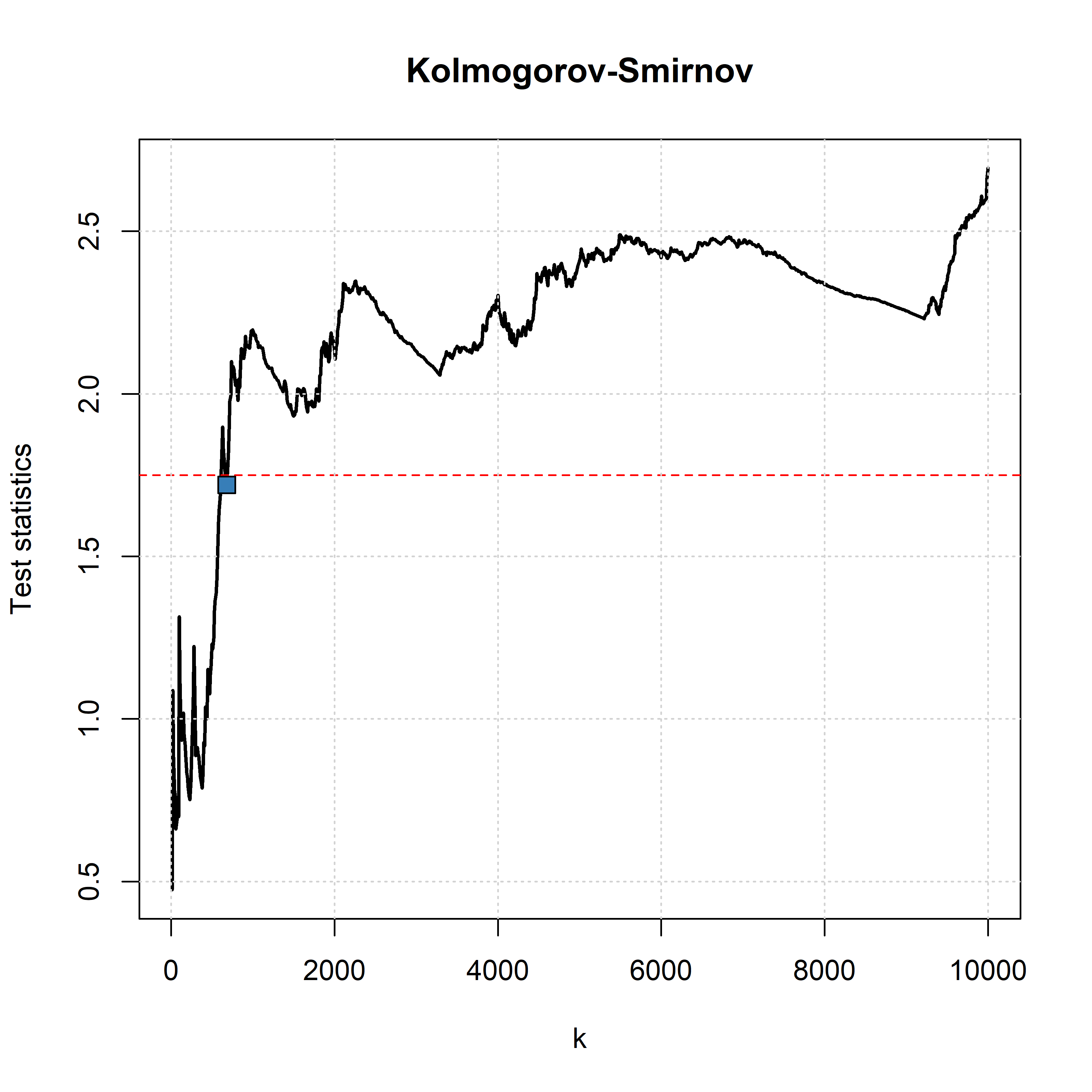}
    \caption{Censored Hill plot with automatically selected sample fractions (top panel) together with the corresponding Extreme Cram\'er--von Mises (center panel) and Extreme Kolmogorov--Smirnov (bottom panel) statistics. The red dotted lines indicate the bounds $L$.} 
    \label{fig:data}
\end{figure}

\begin{figure}[htbp!]
    \centering
        \includegraphics[width=0.45\textwidth]{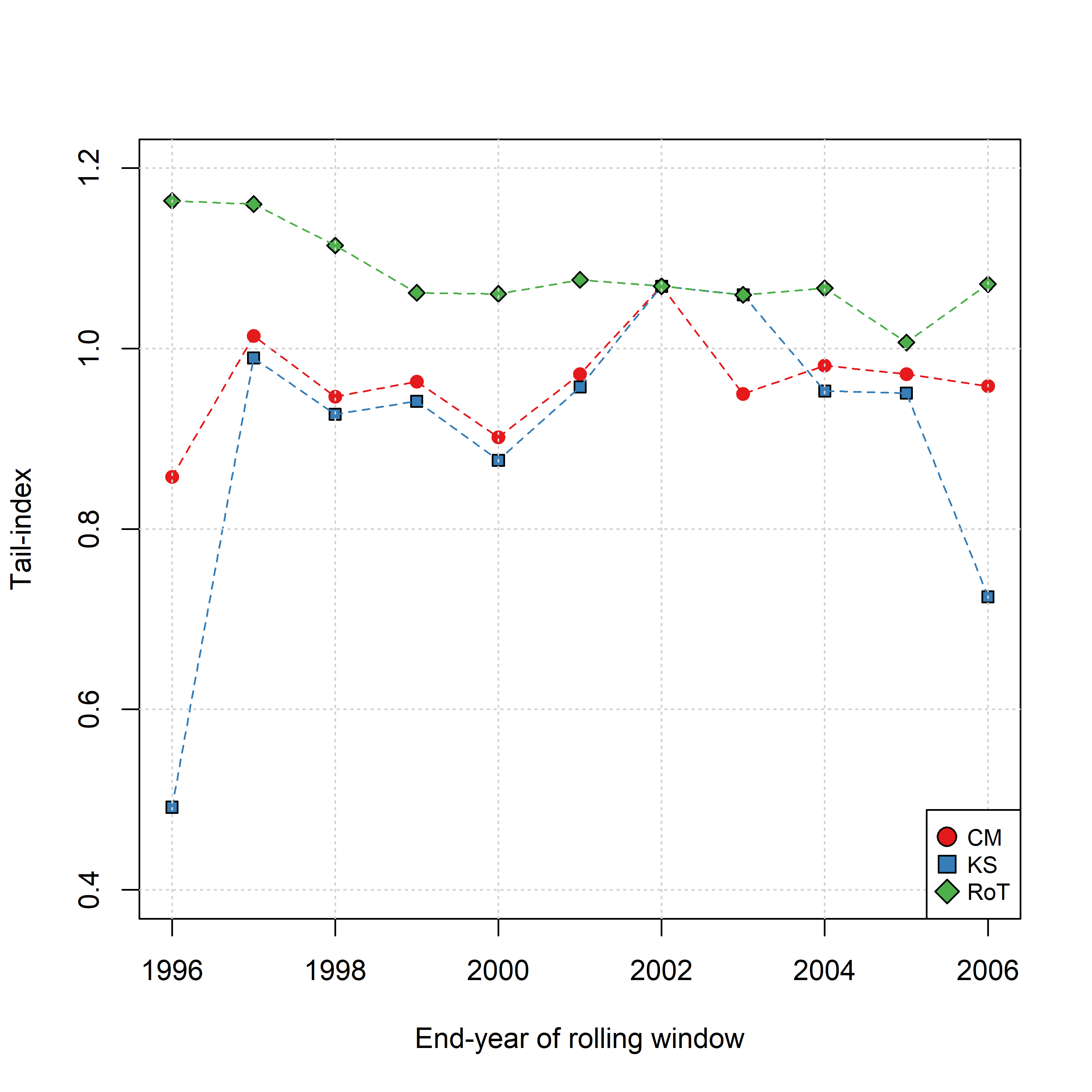}
        \includegraphics[width=0.45\textwidth]{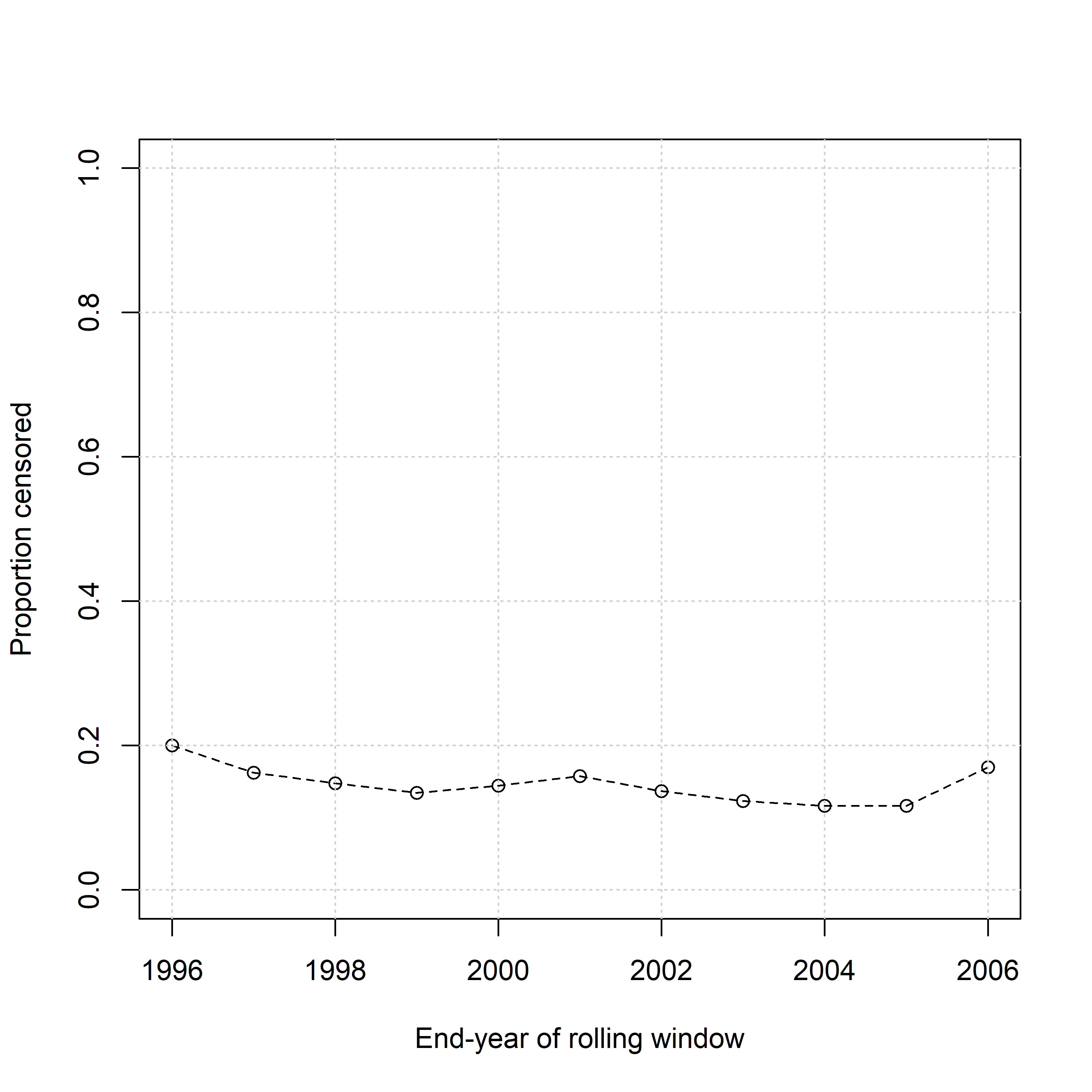}
    \caption{Automatically selected tail estimates for a rolling window based on different selection rules.} 
    \label{fig:data-rolling}
\end{figure}

We further investigate a growing window approach. We thus begin by considering claims arriving from 1992 to 1996. Subsequently, we expand this window to include claims arriving from 1992 to 1997, then 1992 to 1998, and so forth. With this method we would expect that the censoring proportion is decreasing as we include more years. This enable us to see how our estimates from the selection rules behave as the data becomes less and less censored. The estimated tail-index based the the different selection rules and the censoring proportion can be seen in Figure \ref{fig:data-growing}. Again, we observe that both of our proposed approaches suggest a lower tail index than taking the top $1/5$ of the data, indicating that there is a significant bias term in the data. Overall, the rules evolve similarly; however, the last two years, both CM and KS decrease, where RoT remains around the same level. This might indicate that there could be tempering in the far tail of the distribution, possibly related to the data generating process, but also could arise from strategic changes in the handling of claim payments. These effects may not be distinguishable from a rolling window approach, but only apparent with the much larger sample sizes available from the growing window approach. A decreasing trend could be handled by a modification of our proposed models allowing the inclusion of time as a covariate, for which Kernel methods play a central role. The latter is out of scope of the current paper and left for subsequent research.

\begin{figure}[htbp!]
    \centering
        \includegraphics[width=0.45\textwidth]{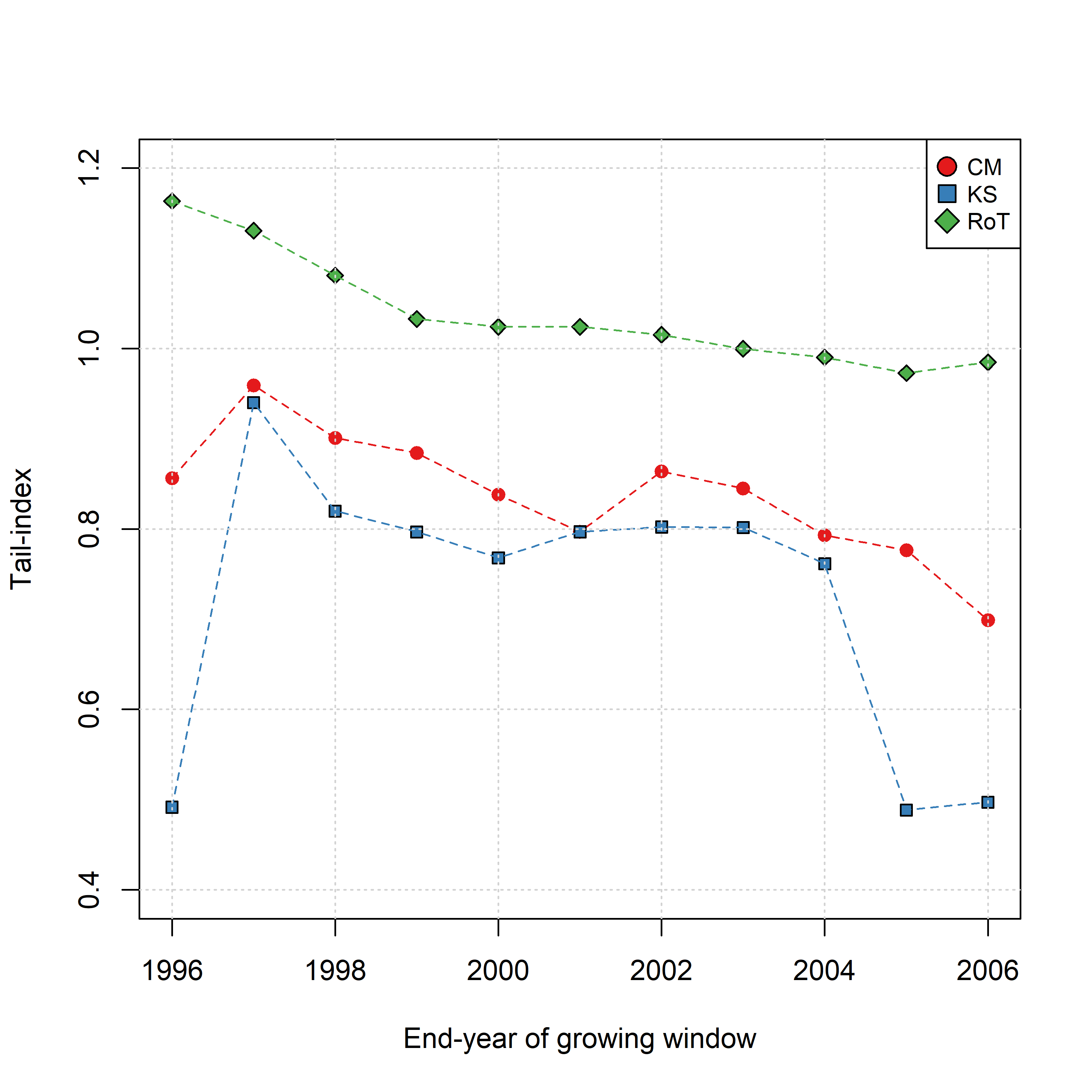}
        \includegraphics[width=0.45\textwidth]{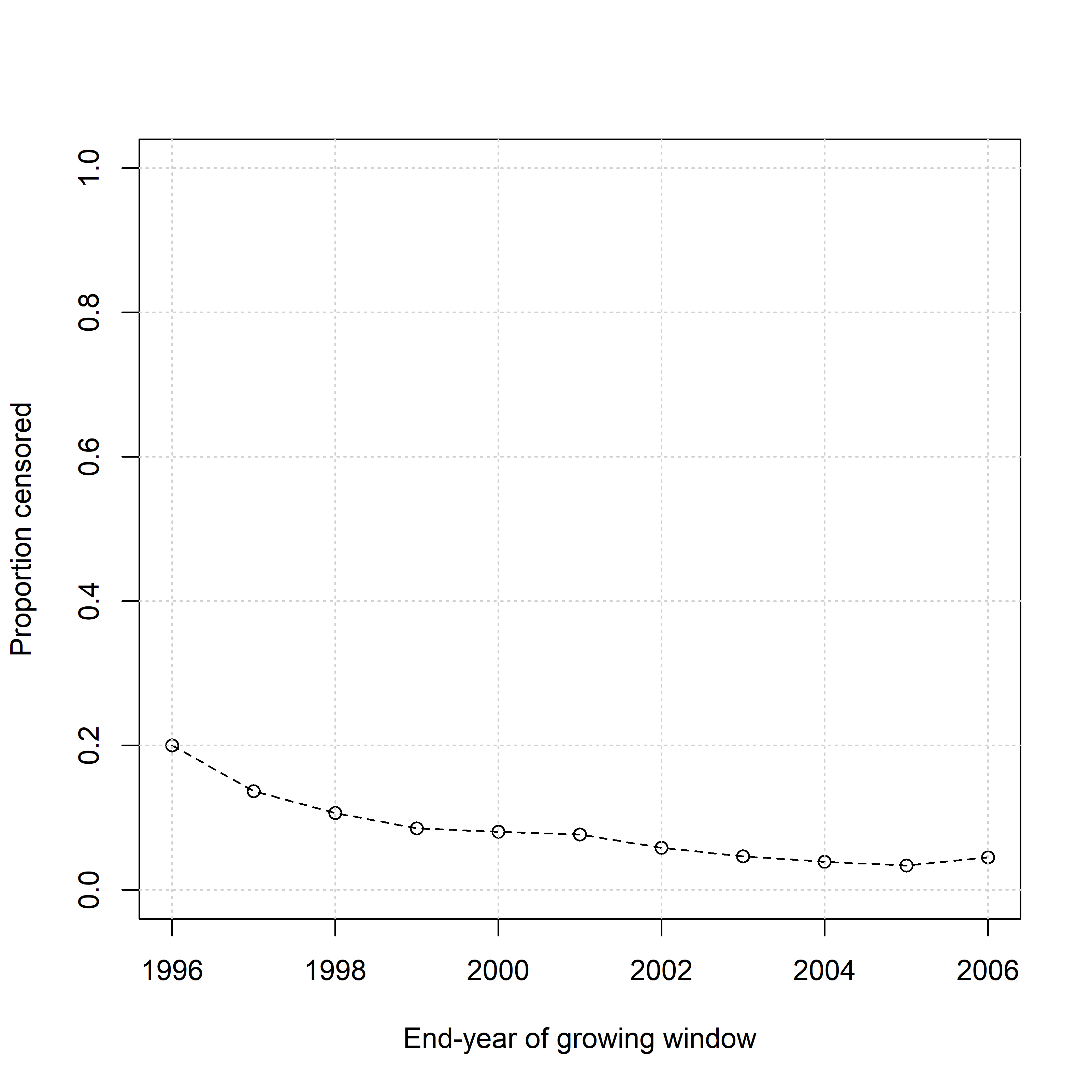}
    \caption{Automatically selected tail estimates for a growing window based on different selection rules.} 
    \label{fig:data-growing}
\end{figure}

\section{Conclusion}

In this paper, we have established pathwise convergence for the Extreme Nelson-Aalen (ENA) and Extreme Kaplan-Meier (EKM) estimators. This was achieved by first demonstrating the convergence of the tail empirical process and then applying the continuous mapping theorem or functional delta method. Using the pathwise convergence of the EKM estimator, we established the consistency and normality of the censored Hill estimator. Additionally, we employed the pathwise convergence to derive the asymptotic distributions of the Goodness-of-Fit (GoF) statistics -- the Kolmogorov--Smirnov and Cramér-von Mises statistics -- in a censored extreme setting. These statistics were subsequently used to construct data-based selection rules.

The two selection rules based on GoF testing were investigated through simulations and real data application. Our proposed selection rules demonstrated favorable performance when the degree of censoring is not severe, and for sufficiently large sample sizes.

\section*{Funding}
The work presented here is supported by the Carlsberg Foundation, grant CF23-1096.

\newpage
\bibliographystyle{plainnat}
\bibliography{GOF_extremes.bib}

\end{document}